\pdfoutput=1
\documentclass[11pt]{article}

\newif\ifjcn
\jcnfalse

\usepackage[american]{babel}
\usepackage[utf8]{inputenc}
\usepackage[T1]{fontenc}
\usepackage{lmodern}
\usepackage{microtype}
\usepackage{amsmath,amssymb,amsthm}
\ifjcn
  \usepackage[a4paper,margin=25mm]{geometry}
\else
  \usepackage[letterpaper,margin=1in]{geometry}
\fi
\usepackage{graphicx}
\usepackage{adjustbox}
\usepackage{booktabs}
\usepackage[font=small,labelfont=bf,skip=6pt]{caption}
\usepackage{enumitem}
\usepackage{cite}  
\usepackage{tikz}
\usetikzlibrary{arrows,decorations.markings,shapes.geometric,shapes,calc,
                positioning,fit,backgrounds,arrows.meta,decorations.pathreplacing,
                patterns,matrix,decorations.pathmorphing}
\definecolor{s1col}{RGB}{31,105,175}
\definecolor{s2col}{RGB}{205,105,15}
\definecolor{icol}{RGB}{198,28,58}
\definecolor{gcol}{RGB}{18,135,78}
\definecolor{obscol}{RGB}{115,115,125}
\definecolor{lcol}{RGB}{128,64,160}
\tikzset{
  vtx/.style   = {circle,draw,thick,minimum size=6.5mm,inner sep=0pt,font=\footnotesize},
  v1/.style    = {vtx,draw=s1col,fill=s1col!12},
  v2/.style    = {vtx,draw=s2col,fill=s2col!15},
  vn/.style    = {vtx,draw=obscol,fill=black!4},
  vb/.style    = {double,double distance=0.9pt},
  dead/.style  = {vtx,draw=obscol!40,fill=white,text=obscol!65},
  e/.style     = {-{Stealth[length=2.2mm]},thick},
  eobs/.style  = {e,draw=obscol},
  ei/.style    = {e,draw=icol,dashed,line width=1pt},
  eg/.style    = {e,draw=gcol,line width=1.2pt},
  el/.style    = {e,draw=lcol,densely dotted,line width=1.2pt},
  ex/.style    = {e,draw=obscol!35,densely dotted},
  ew/.style    = {draw=gcol!75,line width=2.6pt,opacity=.35,-},
  blob1/.style = {rounded corners=7pt,draw=s1col!55,fill=s1col!6,thick,inner sep=5pt},
  blob2/.style = {rounded corners=7pt,draw=s2col!55,fill=s2col!7,thick,inner sep=5pt},
  lbl/.style   = {font=\scriptsize},
  panel/.style = {font=\small\bfseries},
}

\usepackage[colorlinks=true, linkcolor=blue!60!black, citecolor=green!40!black,
            urlcolor=blue!60!black]{hyperref}
\hypersetup{
  pdftitle={Emergence in Network Systems is Bounded by Boundary-Crossing Paths},
  pdfauthor={Johnny Jingze Li and Gabriel A. Silva},
  pdfkeywords={emergence, complex networks, interacting systems,
               boundary-crossing paths, attribution, coarse-graining}}
\theoremstyle{plain}
\newtheorem{theorem}{Theorem}[section]
\newtheorem{proposition}[theorem]{Proposition}
\newtheorem{lemma}[theorem]{Lemma}
\newtheorem{corollary}[theorem]{Corollary}
\theoremstyle{definition}
\newtheorem{definition}[theorem]{Definition}
\newtheorem{example}[theorem]{Example}
\theoremstyle{remark}
\newtheorem{remark}[theorem]{Remark}

\newcommand{\prodedge}[2][\Phi]{#1\langle #2\rangle}

\newcommand{\alttext}[1]{\ifjcn\par\smallskip
  {\leftskip=0pt\rightskip=0pt\parfillskip=0pt plus 1fil
   \noindent\footnotesize\textit{Alt text:} #1\par}\fi}

\newcommand*{\coderepo}{}
\newcommand{\codestatement}{%
  \ifjcn
    \ifx\coderepo\empty as supplementary material\else at \expandafter\url\expandafter{\coderepo}\fi
  \else
    as ancillary files with the arXiv version of this paper%
    \ifx\coderepo\empty\else{} and at \expandafter\url\expandafter{\coderepo}\fi
  \fi}

\title{\bfseries Emergence in Network Systems is Bounded by\\
       Boundary-Crossing Paths}
\author{Johnny Jingze Li\thanks{Department of Mathematics and Center
for Engineered Natural Intelligence, UC San Diego, La Jolla, CA, USA.
Email: \texttt{jil164@ucsd.edu}.}
\and Gabriel A.\ Silva\thanks{Department of Bioengineering, Department of
Neurosciences, and Center for Engineered Natural Intelligence, UC San Diego,
La Jolla, CA, USA.}}
\date{}

\begin{document}
\maketitle

\begin{abstract}
Emergence, features of a whole that no part shows alone, is central to complex
systems but rarely measured so as to locate its source. Viewing a system
structurally, as a network of interacting parts, we evaluate emergence as what
observing the coupled whole reveals beyond, or erases from, the combined
observations of the parts. For observations that build what they see route by
route, we show that every emergent feature is produced by its own route across
the interface between parts, a route that traces it to the components and
interactions responsible. We show that the number of all such boundary-crossing
paths bounds the total emergence discrepancy and equals the emergence of measures
that count routes. The bound is computed locally around the interface and
validated through network examples; in a nervous system, the routes show that interneurons relay most of the emergence.
Our work associates emergence with path structure in the network system that
can be anticipated, attributed, and engineered.


\end{abstract}

\section{Introduction}\label{sec:intro}

Couple two systems and something new can appear. Interdependent networks fail
in abrupt cascades that neither network shows on its
own~\cite{buldyrev2010catastrophic}, interacting neurons carry out collective
computations~\cite{hopfield1982neural}, and language models acquire abilities
at scale that smaller models lack~\cite{wei2022emergent}. Each case raises two
questions. \emph{Which interactions are responsible for an emergent phenomenon?
And how much emergence can a given coupling create?} This paper addresses both
from the structure of the system. For observations that build what they see
route by route, emergence is produced along the routes that cross the
interface between interacting parts: each unit of emergence is traced to a
route of its own, and the number of these routes bounds how much emergence the
coupling can create (Figure~\ref{fig:overview}).

\begin{figure}[tbp]
\centering
\includegraphics[width=\textwidth]{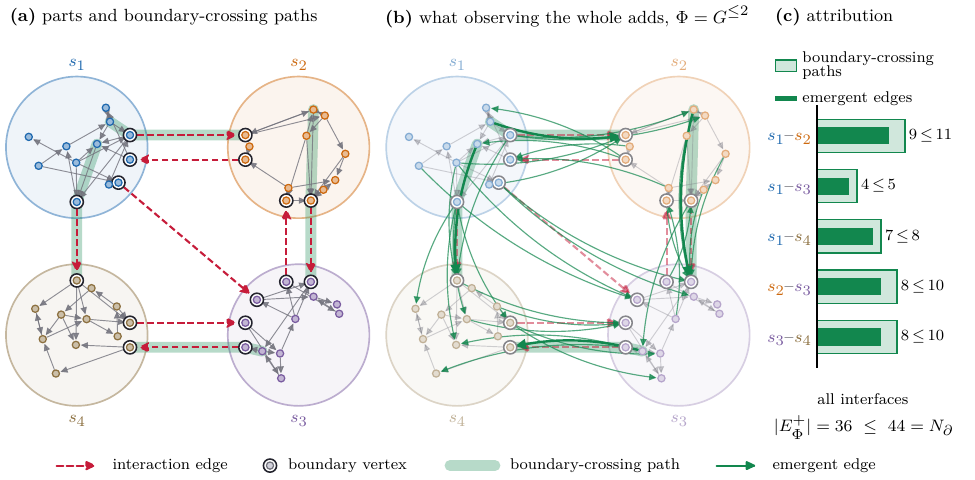}
\caption{\textbf{Emergence travels along the paths that cross between
interacting parts.}
\textbf{(a)}~Four modules $s_1,\dots,s_4$ of the same kind, small random
directed networks, are coupled by eight interaction edges (dashed red); the
ringed endpoints of these edges form the interaction boundary.  Green tracks
mark four boundary-crossing paths, paths that traverse an interaction edge.
\textbf{(b)}~Observing the whole by two-step reachability
($\Phi=G^{\le2}$) reveals $36$ emergent edges (green): edges that are absent
when each module is observed alone and the observed modules are then coupled.
Each is produced by a boundary-crossing path; the four bold ones are produced
by the paths marked in~(a).
\textbf{(c)}~For each coupled pair of modules, the emergent edges produced
across their interface (dark) are at most as many as the boundary-crossing
paths through it (light); no path here crosses two interfaces, and in all
$36$ emergent edges ($|E^+_\Phi|$) stand against $44$ boundary-crossing paths
($N_\partial$), the difference being the eight interaction edges themselves.}
\alttext{Three panels. Left: four softly colored circular regions labeled s1
to s4, each holding a small directed network of the same random kind,
connected by eight red dashed arrows; the vertices at the ends of the red
arrows are ringed, and four wide green bands trace short routes from inside one
region, across a red arrow, into another. Middle: the same network faded, with
thirty-six green arrows, each joining vertices of two regions linked by a red
arrow; four of them are bold and run alongside the green bands. Right: a
horizontal bar chart with one row per connected pair of regions, each showing
a dark bar for emergent edges inside a longer light bar for boundary-crossing
paths, with values 9 and 11, 4 and 5, 7 and 8, 8 and 10, and 8 and 10, and a
total of 36 emergent edges against 44 paths.}
\label{fig:overview}
\end{figure}

\subsection{Emergence from a structural perspective}\label{sec:structural}

That the whole can exceed the sum of its parts is a recurring theme across the
sciences, from new properties at each level of
organization~\cite{anderson1972more} and complex behavior of simple
interacting components~\cite{holland1998emergence} to self-organizing
biology~\cite{kauffman1993origins}, brains~\cite{sporns2010networks,bullmore2009complex},
and learned systems~\cite{ganguli2022predictability}. Quantitative accounts,
mostly information-theoretic or
causal~\cite{tononi2004information,rosas2020reconciling,gershenson2012complexity},
measure from dynamics or statistics how much a whole exceeds its parts, often
by searching over partitions~\cite{oizumi2014phenomenology} or
coarse-grainings~\cite{hoel2013quantifying,klein2020emergence}. We take a
structural perspective: we study the structure of systems that gives rise to
emergence and ask which interaction patterns are responsible for it. Complex
systems are built from nearly decomposable modules, with strong interactions
inside each module and weaker ones between them~\cite{simon1962architecture};
network science has made directed networks a common language for such
architectures~\cite{newman2003structure,albert2002statistical,newman2018networks},
has shown that real networks are
modular~\cite{girvan2002community,newman2006modularity}, and has found that
specific wiring patterns carry out specific
functions~\cite{milo2002network,alon2007network}. If emergence is what
coupling adds, the interface between modules is where to look for its source.

Emergence is also relative to a way of looking at the system: it is defined
with respect to an observer~\cite{crutchfield1994calculi}, and whether a
capability appears emergent depends on the metric used to observe
it~\cite{schaeffer2023emergent}. We therefore make the observation explicit: a
system is a directed network, its parts are subnetworks on disjoint vertex
sets, their \emph{interaction} is the set of cross edges joining them, and an
\emph{observation} turns a network into an observed network, for instance by
linking each vertex to every vertex it reaches in at most $L$ steps.
\emph{Emergence} is the discrepancy between observing the whole and combining
the observations of the parts: the \emph{emergent edges} appear only in the
observed whole, and the \emph{lost edges} only in the observed parts joined by
the interaction edges.

\subsection{The idea: emergence travels along boundary-crossing paths}\label{sec:idea}

Figure~\ref{fig:overview} shows the idea. Four modules are coupled by eight
interaction edges (dashed red), whose endpoints form the interaction boundary.
A \emph{boundary-crossing path} is a short route that traverses at least one
interaction edge. Observing the network by two-step reachability,
$\Phi=G^{\le2}$, reveals $36$ emergent edges, pairs linked within two steps
only through the coupling. Each is produced by its own boundary-crossing path,
its \emph{witness}; of the $44$ boundary-crossing paths, the other eight are the
interaction edges themselves, which produce nothing new.

The reason is simple. A route inside one part is seen when that part is
observed alone; a route reveals something new only by joining two
parts, and it passes from one part to another only along an interaction edge. Behind this is a
change of viewpoint: a network is its paths. The paths of a joined system are
those of each part plus those through an interaction edge, so passing to
paths makes the effect of coupling exactly additive, a linearization that lets
us count emergence route by route.

Our main result, Theorem~\ref{thm:main}, makes this precise for every
observation that \emph{acts on paths}: each route of bounded length either
produces one observed edge or produces nothing, depending on the route alone,
and every observed edge is produced by some route. Examples are $L$-step
reachability ($G^{\le L}$, a graph power), alone or combined with
coarse-graining or thresholding, and deletion rules that read short routes or
cycles. With $E^+_\Phi$ the set of emergent edges and $N_\partial$ the number
of boundary-crossing paths that produce an edge, the theorem states:
\begin{enumerate}[label=(\roman*),itemsep=2pt,topsep=4pt]
\item the emergent edges are exactly the edges that boundary-crossing paths
produce and that the observed parts and the interaction do not already show;
\item distinct emergent edges have distinct witnesses, so that
\[
  |E^+_\Phi|\;\le\;N_\partial,
\]
and the witness map traces each unit of emergence to one route;
\item equality holds exactly when the producing boundary-crossing paths
produce distinct edges that the observed parts and the interaction do not
already show; for graph powers, whose boundary-crossing paths include the
interaction edges themselves, the same criterion, applied to the paths of
length at least two, characterizes equality in the sharpened bound of
Corollary~\ref{cor:sharp}.
\end{enumerate}
The total discrepancy, lost edges included, is at most the number of all
boundary-crossing paths (Corollary~\ref{cor:discrepancy}); for observables
that sum a weight over routes, the same paths give the emergence exactly, and
for the path count it equals their number (Theorem~\ref{thm:scalar} and
Corollary~\ref{cor:pathcount}).

The bound also exposes a multiplicative signature. For graph powers and an
interaction that runs one way, an interaction edge couples every route into
its tail with every route out of its head, so the boundary-crossing paths and
the emergent edges have closed forms in the reach of the two parts
(Proposition~\ref{prop:bridge}). If the tail receives edges from $k$ vertices,
the head sends edges to $m$ vertices, and there are no other edges, three-step
reachability reveals $k+m+km$ emergent edges (Example~\ref{ex:bowtie}). The
term $km$, in-reach times out-reach, lets a single edge between internally
structured modules create a surge: in our experiments, under four-step
reachability, the mean number of emergent edges created by one edge between
two eight-vertex modules rises from $0$ to $39.3$ as their internal mean
out-degree grows from $0$ to $3$, matching the exact expectation $39.0$ given
by the bridge law (Corollary~\ref{cor:bridge_random}).

\subsection{Why this matters}\label{sec:why}

\begin{description}[leftmargin=0pt,itemsep=3pt,topsep=3pt,font=\normalfont\bfseries]
\item[Novel.] To our knowledge, no previous work bounds and attributes
emergence by the routes that cross the interface between parts; we build on
the generative effects of Adam and Dahleh~\cite{adam2019generativity} and on
incremental algorithms for reachability and database
queries~\cite{gupta1993maintaining,italiano1986amortized}, and add this bound
and attribution for a whole class of observations. Where existing measures
compute emergence from dynamics or statistics, and studies of interacting
networks show how coupling reshapes processes on
them~\cite{leicht2009percolation,kivela2014multilayer}, we locate its source in
boundary-crossing paths and count them (Section~\ref{sec:related}).

\item[Mechanistic.] Each unit of emergence is traced to a route, and the
routes that a coupling opens can be counted before the coupled system is
observed. Their number anticipates and bounds the emergence the coupling can
create, approximates it where distinct boundary-crossing paths produce distinct
edges, as in large sparse networks (Sections~\ref{sec:tightness} and~\ref{sec:scale}), and
ranks candidate couplings better than the cut size, the degrees at the
interface, or edge betweenness once the observation reads three or more steps
(Section~\ref{sec:baselines}). Cutting every boundary-crossing path that
produces an edge removes emergence (Proposition~\ref{prop:cut}), and bridging
internally structured modules amplifies it (Sections~\ref{sec:bridge}
and~\ref{sec:bridging}), much as tuning interconnection can
suppress or amplify cascades~\cite{brummitt2012suppressing}. This echoes
mechanistic accounts of abrupt capabilities in learned
systems~\cite{nanda2023progress,olsson2022incontext}.

\item[Attribution.] The boundary-crossing paths identify which components and
interactions are responsible for each emergent feature: each feature is shared
equally among the routes that produce it, and each vertex and edge receives
the shares of the routes through it (Definition~\ref{def:shares}). Grouped by
interface, these routes show how much of the emergence each coupling carries
(Figure~\ref{fig:overview}(b)), and removing
interaction edges removes exactly the features whose routes all use them
(Proposition~\ref{prop:intervene}), a structural counterpart to circuit
discovery in learned
networks~\cite{olah2020zoom,wang2023interpretability,conmy2023automated}. In
the \emph{C.~elegans} connectome, interneurons, $30\%$ of the neurons, relay
$66\%$ of the emergence under two-step reachability, and cutting their synapses
onto motor neurons removes, as predicted, $38\%$ of the emergent edges
(Section~\ref{sec:real}).

\item[Computable.] The bound is local. The boundary-crossing paths are found by
a backward and a forward search from each interaction edge, like the searches
that accumulate shortest-path counts for betweenness~\cite{brandes2001faster},
at a cost that, for a given reach and maximum degree, is independent of the
size of the parts (Corollary~\ref{cor:local}). On sparse parts of up to $10^5$
vertices joined by up to $16$ interaction edges, the local search returns the exact emergent edges of a graph power
in milliseconds, without applying the observation to the whole system, which
takes seconds (Section~\ref{sec:scale}). A bound by walks, cheaper for dense
interfaces, is computed by sparse matrix--vector
products~\cite{biggs1993algebraic,newman2018networks}
(Appendix~\ref{app:computation}).

\item[Grounded.] Working with paths loses nothing: a directed network and the
set of its paths carry the same information, as the path algebras of quivers
(directed graphs) make
precise~\cite{gabriel1972unzerlegbare,assem2006elements,schiffler2014quiver,derksen2017introduction}.
Our prior work quantified emergence in that setting~\cite{li2025categorical};
here we work with the paths directly (Section~\ref{sec:paths}).
\end{description}

\subsection{Contributions and organization}

\begin{enumerate}[label=(\arabic*),leftmargin=2em,itemsep=1pt,topsep=3pt]
\item A structural framework in which emergence is a discrepancy
(Sections~\ref{sec:framework} and~\ref{sec:discrepancy}).
\item The main theorem (Theorem~\ref{thm:main}): the emergent edges are exactly
the new edges produced by boundary-crossing paths, at most $N_\partial$ of
them, each attributed to its own witness, with equality characterized;
corollaries on the total discrepancy, sharpening, locality, and cost, and
their extension to many parts; shares that divide each emergent edge among the paths that produce it;
a closed form for a single bridge, in which the emergent edges pair the
in-reach of its tail with the out-reach of its head; and an exact identity for
observables that sum over paths (Section~\ref{sec:main}).
\item What produces emergence, what suppresses it, and the exact effect of
removing interaction edges (Section~\ref{sec:examples}).
\item Exact verification on $8935$ instances, with parts from sixteen
synthetic and empirical network families, and experiments on tightness,
topology, reach, granularity, bridging, parts of up to $10^5$ vertices,
simpler statistics of the coupling, and three real directed networks
(Section~\ref{sec:numerics}).
\item A bound on the emergence potential of a subsystem in a larger network by
the walks through its boundary, related to truncated Katz centrality
(Appendix~\ref{app:potential}).
\end{enumerate}
Section~\ref{sec:discussion} discusses related work, applications, and open
questions; Appendices~\ref{app:computation} and~\ref{app:experiments} treat
computation and the design of the experiments.

\section{A structural framework}\label{sec:framework}

We describe a system by its structure, a directed network of components and
their interactions~\cite{newman2003structure,newman2018networks}, with three
ingredients: parts joined by an interaction, an observation that turns a
network into what is seen of it, and paths.

\subsection{Systems, interaction, and boundary}

All graphs are finite, directed, and loopless~\cite{bang2009digraphs}, with
vertices in one fixed finite set $\Omega$, so that graphs built from different
parts are comparable edge by edge.  A graph $G=(V(G),E(G))$ has
$V(G)\subseteq\Omega$ and $E(G)\subseteq V(G)\times V(G)$, and $G\subseteq H$
means $V(G)\subseteq V(H)$ and $E(G)\subseteq E(H)$.  When edges carry weights,
each ordered pair of vertices has one fixed weight, which every subgraph
inherits.

\begin{definition}[Parts, interaction, join, and boundary]\label{def:interaction}
Let $s_1=(V_1,E_1)$ and $s_2=(V_2,E_2)$ be graphs on disjoint vertex sets, the
\emph{parts}.  An \emph{interaction} is a set of cross edges
$\mathcal{I}\subseteq(V_1\times V_2)\cup(V_2\times V_1)$, and the \emph{join}
of the parts along $\mathcal{I}$ is the graph
\[
s_1\vee_{\mathcal{I}}s_2:=\bigl(V_1\cup V_2,\;E_1\cup E_2\cup\mathcal{I}\bigr),
\]
written $s_1\vee s_2$ when $\mathcal{I}$ is clear.  The endpoints of the
interaction edges form the \emph{interaction boundary}
$\partial\subseteq V_1\cup V_2$.
\end{definition}

The join keeps each part intact and adds exactly the prescribed coupling, in
either direction; in multilayer terms, the parts are layers and the
interaction edges are interlayer edges~\cite{kivela2014multilayer}.

Figure~\ref{fig:interaction}(a) shows two parts, their interaction, and its
boundary.  The boundary is where the parts touch: a route from one part to the
other passes along an interaction edge.  The main results, listed in
Remark~\ref{rem:multipart}, apply verbatim to any number of parts
$s_1,\dots,s_k$ on pairwise disjoint vertex sets, with the interaction a set
of edges between different parts, as in networks of networks~\cite{gao2012networks}.

\begin{figure}[htbp]
\centering
\begin{tikzpicture}[x=10mm,y=11mm,baseline=(current bounding box.north)]
  \node[v1]    (a1) at (0,1)   {$a_1$};
  \node[v1]    (a2) at (0,0)   {$a_2$};
  \node[v1,vb] (b1) at (1.25,1) {$b_1$};
  \node[v1,vb] (b2) at (1.25,0) {$b_2$};
  \draw[eobs] (a1) -- (b1);
  \draw[eobs] (a1) -- (a2);
  \draw[eobs] (a2) -- (b2);
  \node[v2,vb] (c1) at (3.35,1) {$c_1$};
  \node[v2,vb] (c2) at (3.35,0) {$c_2$};
  \node[v2]    (d1) at (4.6,1)  {$d_1$};
  \node[v2]    (d2) at (4.6,0)  {$d_2$};
  \draw[eobs] (c1) -- (d1);
  \draw[eobs] (d1) -- (d2);
  \draw[eobs] (d2) -- (c2);
  \begin{scope}[on background layer]
    \node[blob1,fit=(a1)(b2)] (S1) {};
    \node[blob2,fit=(c1)(d2)] (S2) {};
  \end{scope}
  \node[s1col,font=\small,anchor=south] at (S1.north) {$s_1$};
  \node[s2col,font=\small,anchor=south] at (S2.north) {$s_2$};
  \node[panel,anchor=south west] at ($(S1.north west)+(-1mm,0)$) {(a)};
  \draw[ei] (b1) -- (c1);
  \draw[ei] (c2) -- (b2);
  \node[icol,font=\small] at (2.3,0.5) {$\mathcal{I}$};
  \draw[decorate,decoration={brace,amplitude=5pt,mirror},obscol,thick]
    ($(b2.west)+(0,-7mm)$) -- ($(c2.east)+(0,-7mm)$)
    node[midway,below=5pt,font=\footnotesize,text=black]
    {interaction boundary $\partial=\{b_1,b_2,c_1,c_2\}$};
\end{tikzpicture}\hspace{7mm}%
\begin{tikzpicture}[x=11mm,y=11mm,baseline=(current bounding box.north),
    band/.style = {line width=7pt,line cap=round,line join=round,opacity=.3},
    hl1/.style  = {band,draw=s1col},
    hl2/.style  = {band,draw=s2col},
    hlg/.style  = {band,draw=gcol,opacity=.38},
  ]
  \node[v1]    (a) at (0,1) {};
  \node[v1]    (b) at (1,1) {};
  \node[v1]    (c) at (2,1) {};
  \node[v1]    (g) at (1,0) {};
  \node[v1,vb] (h) at (2,0) {};
  \draw[eobs] (a) -- (b); \draw[eobs] (b) -- (c);
  \draw[eobs] (g) -- (h); \draw[eobs] (a) -- (g);
  \node[v2]    (d) at (3.6,1) {};
  \node[v2]    (e) at (4.6,1) {};
  \node[v2]    (f) at (5.6,1) {};
  \node[v2,vb] (i) at (3.6,0) {};
  \node[v2]    (j) at (4.6,0) {};
  \draw[eobs] (d) -- (e); \draw[eobs] (e) -- (f);
  \draw[eobs] (i) -- (j); \draw[eobs] (j) -- (f);
  \draw[ei] (h) -- (i);
  \begin{scope}[on background layer]
    \node[blob1,fit=(a)(c)(h)] (S1) {};
    \node[blob2,fit=(d)(f)(i)] (S2) {};
    \draw[hl1] (a.center) -- (b.center) -- (c.center);
    \draw[hl2] (d.center) -- (e.center) -- (f.center);
    \draw[hlg] (g.center) -- (h.center) -- (i.center) -- (j.center);
  \end{scope}
  \node[s1col,font=\small,anchor=south] at (S1.north) {$s_1$};
  \node[s2col,font=\small,anchor=south] at (S2.north) {$s_2$};
  \node[panel,anchor=south west] at ($(S1.north west)+(-1mm,0)$) {(b)};
  \newcommand{\crosskeyband}[1]{\tikz[baseline=-0.55ex]\draw[#1](0,0)--(5.5mm,0);}
  \node[font=\footnotesize,anchor=north] at ($(S1.south west)!.5!(S2.south east)+(0,-2.5mm)$)
    {\crosskeyband{hl1}\ path of $s_1$\qquad \crosskeyband{hl2}\ path of $s_2$\qquad
     \crosskeyband{hlg}\ new path};
\end{tikzpicture}
\caption{\textbf{Parts, interaction, and boundary-crossing paths
(Definition~\ref{def:interaction} and
Lemma~\ref{lem:crossing}).}
\textbf{(a)}~Parts $s_1$ (blue) and $s_2$ (orange) on disjoint vertex sets,
their interaction $\mathcal{I}$ (dashed red, in either direction), and its
boundary $\partial$ (double rings).
\textbf{(b)}~A path that avoids $\mathcal{I}$ stays in one part (blue,
orange); the paths of $s_1\vee s_2$ that belong to neither part are exactly
those that traverse an interaction edge (green).}
\alttext{Two panels. Left: two rounded regions of circles, blue on the left
labeled s1 and orange on the right labeled s2, each containing gray arrows
between its own vertices. Two dashed red arrows join the regions, one
pointing right and one pointing left; the four circles they touch have
double outlines, and a brace beneath them labels them as the interaction
boundary. Right: two similar regions of unlabeled circles joined by a single
dashed red arrow from a double-ringed blue circle to a double-ringed orange
circle. A blue band highlights a path of three circles inside the blue
region, an orange band highlights a path inside the orange region, and a
green band highlights a path that starts in the blue region, crosses the red
arrow, and continues in the orange region.}
\label{fig:interaction}\label{fig:crossing}
\end{figure}

\subsection{Observations}

Emergence is relative to a way of looking at a system
(Section~\ref{sec:structural}).  We formalize a way of looking as a graph
operator~\cite{prisner1995graph}, a rule that turns each graph into another,
such as a graph power or a coarse-graining, together with what the rule does
to vertices, since some observations merge them.

\begin{definition}[Observation]\label{def:obs}
An \emph{observation} is a rule $\Phi$ that assigns to every graph $G$ on
$\Omega$ a graph $\Phi(G)$ on a second fixed set $\Omega'$, together with a
\emph{vertex map} $\varphi:\Omega\to\Omega'$, fixed in advance and the same for
every $G$ (for a coarse-graining, $\varphi$ sends each vertex to its block).
Unless the observation merges vertices, $\Omega'=\Omega$ and $\varphi$ is the
identity.  The observation is \emph{monotone} if $G\subseteq H$ implies
$E(\Phi(G))\subseteq E(\Phi(H))$.
\end{definition}

The following examples are all monotone; Figure~\ref{fig:obsmaps} shows
several of them on one small graph.
\begin{enumerate}[label=(\roman*),leftmargin=*,itemsep=1pt,topsep=3pt]
\item \emph{Identity}: $\Phi(G)=G$.
\item \emph{Graph power} $G^{\le L}$: the same vertices, with an edge
$(u,v)$, $u\ne v$, whenever a path of length at most $L$ runs from $u$ to
$v$~\cite{bang2009digraphs}; it records $L$-step reachability.
\item \emph{Thresholding} at a fixed $\tau$: keep the edges of weight at least
$\tau$.
\item \emph{Coarse-graining}~\cite{song2005selfsimilarity,itzkovitz2005coarse}
by a partition of $\Omega$ fixed in advance: the blocks that meet $V(G)$, with
an edge $(B,C)$, $B\ne C$, whenever an edge of $G$ runs from $B$ to $C$.
\item \emph{Deletion rules that read context}: \emph{drop sinks} deletes the
vertices without an out-edge, so an edge survives exactly when its head has an
out-edge; \emph{edges on a cycle} keeps the edges that lie on a directed cycle
of length at most $L$ (all cycles when $L=|\Omega|$); the \emph{continuation
filter} keeps the edges that begin a path of length at least $L$.
\item \emph{Two compositions}: coarse-grained powers ($G^{\le L}$ followed by
coarse-graining) and thresholded powers (thresholding followed by $G^{\le L}$).
\end{enumerate}
The identity, thresholding, and coarse-graining are \emph{edgewise}: they
decide each observed edge from a single edge of $G$.  The others read routes of
length at least two.  Among these examples, only
the route-reading ones produce emergence (Section~\ref{sec:examples}).
Observations that return a number are treated at the end of
Section~\ref{sec:main} (Definition~\ref{def:discrepancy_scalar}).

\begin{figure}[htbp]
\centering
\begin{tikzpicture}[x=11mm,y=11mm,
    vn/.append style   = {minimum size=5.3mm},
    dead/.append style = {minimum size=5.3mm},
    wt/.style   = {font=\scriptsize,inner sep=1pt,outer sep=0.6pt},
    wtx/.style  = {wt,text=obscol!60},
    desc/.style = {font=\scriptsize,align=center,anchor=north},
    blk/.style  = {vn,rectangle,rounded corners=2.65mm,minimum height=5.3mm,
                   minimum width=5.3mm,inner xsep=1.2pt},
    ttl/.style  = {font=\footnotesize\bfseries,anchor=base},
  ]
  \begin{scope}
    \node[ttl] at (1,1.62) {source $G$};
    \node[vn] (p) at (0,1) {$p$}; \node[vn] (q) at (1,1) {$q$};
    \node[vn] (r) at (2,1) {$r$}; \node[vn] (t) at (0,0) {$t$};
    \draw[eobs] (p) -- node[wt,above]{$0.8$} (q);
    \draw[eobs] (q) -- node[wt,above]{$0.7$} (r);
    \draw[eobs] (t) -- node[wt,right]{$0.2$} (p);
    \node[desc] at (1,-0.36) {weights, read\\by thresholding};
  \end{scope}
  \begin{scope}[xshift=32mm]
    \node[ttl] at (1,1.62) {graph power $G^{\le 2}$};
    \node[vn] (p) at (0,1) {$p$}; \node[vn] (q) at (1,1) {$q$};
    \node[vn] (r) at (2,1) {$r$}; \node[vn] (t) at (0,0) {$t$};
    \draw[eobs] (p) -- (q); \draw[eobs] (q) -- (r); \draw[eobs] (t) -- (p);
    \draw[eg] (p) to[bend left=34] (r);
    \draw[eg] (t) -- (q);
    \node[desc,text=gcol] at (1,-0.36) {creates $(p,r)$\\and $(t,q)$};
  \end{scope}
  \begin{scope}[xshift=64mm]
    \node[ttl] at (1,1.62) {coarse-graining};
    \node[blk,minimum width=16.3mm] (P) at (0.5,1) {$\{p,q\}$};
    \node[blk] (R) at (2,1) {$\{r\}$};
    \node[blk] (T) at (0,0) {$\{t\}$};
    \draw[eobs] (P) -- (R);
    \draw[eobs] (T) -- (T |- P.south);
    \node[desc] at (1,-0.36) {blocks become vertices;\\ $(p,q)$ is absorbed};
  \end{scope}
  \begin{scope}[xshift=96mm]
    \node[ttl] at (1,1.62) {thresholding, $\tau=0.5$};
    \node[vn] (p) at (0,1) {$p$}; \node[vn] (q) at (1,1) {$q$};
    \node[vn] (r) at (2,1) {$r$}; \node[vn] (t) at (0,0) {$t$};
    \draw[eobs] (p) -- node[wt,above]{$0.8$} (q);
    \draw[eobs] (q) -- node[wt,above]{$0.7$} (r);
    \draw[ex]   (t) -- node[wtx,right]{$0.2$} (p);
    \node[desc] at (1,-0.36) {deletes the\\light edge $(t,p)$};
  \end{scope}
  \begin{scope}[xshift=128mm]
    \node[ttl] at (1,1.62) {drop sinks};
    \node[vn]   (p) at (0,1) {$p$}; \node[vn] (q) at (1,1) {$q$};
    \node[dead] (r) at (2,1) {$r$}; \node[vn] (t) at (0,0) {$t$};
    \draw[eobs] (p) -- (q); \draw[ex] (q) -- (r); \draw[eobs] (t) -- (p);
    \node[desc] at (1,-0.36) {deletes the sink $r$\\and $(q,r)$};
  \end{scope}
\end{tikzpicture}
\caption{\textbf{Several observations of one small graph
(Definition~\ref{def:obs}).}  Each panel to the right of the source $G$ applies one
observation to it.  Green edges are created, and faded elements are deleted;
coarse-graining uses the fixed partition $\{p,q\},\{r\},\{t\}$.}
\alttext{Five small panels in one row: the source graph on the four vertices
p, q, r, and t, followed by four observations of it. The source graph has
arrows p to q, q to r, and t to p, labeled with weights 0.8, 0.7, and 0.2.
The graph-power panel adds two green
arrows, p to r and t to q. The coarse-graining panel merges p and q into one
rounded block with arrows from block t into it and from it to block r. The
thresholding panel fades the light arrow from t to p. The drop-sinks panel
fades vertex r and the arrow from q to r.}
\label{fig:obsmaps}
\end{figure}

\subsection{Paths as the basic objects}\label{sec:paths}

A \emph{path} of length $k\ge1$ in $G$ is a sequence $p=(v_0,\dots,v_k)$ of
distinct vertices with $(v_{i-1},v_i)\in E(G)$ for $i=1,\dots,k$; the path
\emph{traverses} these $k$ edges.  A \emph{cycle} of length $k\ge2$ is defined
in the same way with $v_k=v_0$ and otherwise distinct vertices; it is read from
its starting vertex, so a cycle of length $k$ appears $k$ times, once per
rotation.  We fix a length bound $L$ and write $\mathcal{P}_L(G)$ for the set of
paths of $G$ of length at most $L$; for observations that read a return route,
such as edges on a cycle, $\mathcal{P}_L(G)$ and the word ``path'' also include
the cycles of length at most $L$.  The length bound $L$ sets the reach of the observation, for
example the number of layers of a feedforward system, and keeps the counts of
Section~\ref{sec:main} local to the interface.  We take paths, which are simple by
definition, as the basic objects of the theory, for four reasons.

\paragraph{Linearization.}  A graph is a nonlinear object: reachability and
cycles depend on how edges combine, so observing a union can reveal more than
the union of the observations.  Paths restore additivity.  The paths of a
joined system fall into three disjoint classes: the paths of $s_1$, the paths
of $s_2$, and the paths that traverse an interaction edge
(Lemma~\ref{lem:crossing} and Figure~\ref{fig:crossing}(b)).  The effect of
coupling thus becomes an account over individual routes, in the tradition of
network measures built from walks and
paths~\cite{katz1953new,borgatti2006graph}.

\paragraph{Mechanism.}  A path that traverses an interaction edge is an
explicit route from structure in one part to a site in another, where the
observation acts; each counted path is thus a concrete mechanism.

\paragraph{Grounding.}  A directed network and the set of its paths carry the
same information: the edges are the paths of length one, and longer paths are
chains of edges.  The representation theory of quivers (directed graphs), which
goes back to Gabriel~\cite{gabriel1972unzerlegbare}, makes this precise.  The
path algebra of a quiver has its paths as a basis: a trivial path at each
vertex and every walk, a head-to-tail sequence of edges that may revisit
vertices (our simple paths among them), with concatenation as
multiplication; the representations of the quiver, a vector space at each
vertex and a linear map on each edge, correspond to the modules, in the
algebraic sense, over this algebra~\cite{assem2006elements,schiffler2014quiver,derksen2017introduction}.  A
representation acts along a path by composing its maps, so paths are the
building blocks of linear models on networks; a neural network is such a
representation with activation functions~\cite{armenta2021representation}.  Our prior work quantified emergence in this
setting~\cite{li2025categorical}; here we keep its combinatorial core, the
paths.

\paragraph{Simple paths.}  Which kind of route an observable should count
depends on how influence travels~\cite{borgatti2005centrality}.  For routes
between parts we take the simple path, since a walk that returns to a vertex
travels along a route it has already opened; in a feedforward network, such
as a layered neural network, every walk is already a path.

\section{Emergence as the gap between the whole and its parts}\label{sec:discrepancy}

There are two ways to go from the parts to an observation of the whole
(Figure~\ref{fig:square}).  One joins the parts and then observes the joined
system, which gives the \emph{observed whole} $\Phi(s_1\vee s_2)$.  The other observes each part and then
joins the observed parts with the same interaction.  Emergence is the gap
between the two results: the structure seen in the whole beyond what the
observed parts and their interaction show.  For graphs coupled by an explicit
interaction, this gap is a quantitative form of generative
effects~\cite{adam2019generativity,adam2017thesis}.

\begin{definition}[Observed union]\label{def:observed_union}
The \emph{observed union} of $s_1$ and $s_2$ under $\Phi$ is the graph
\[
\Phi(s_1)\vee\Phi(s_2):=\Bigl(V(\Phi(s_1))\cup V(\Phi(s_2))\cup\varphi(\partial),\;
E(\Phi(s_1))\cup E(\Phi(s_2))\cup\varphi(\mathcal{I})\Bigr),
\]
where $\varphi(\mathcal{I}):=\{(\varphi(u),\varphi(v)):(u,v)\in\mathcal{I},\
\varphi(u)\ne\varphi(v)\}$ is the interaction carried to the observed level by
the vertex map (an interaction edge inside one block of a coarse-graining is
dropped).  The interaction is formed \emph{after} the observation, so it
always appears in the observed union, whether or not $\Phi$ keeps it in the
whole.
\end{definition}

When $\varphi$ is the identity, $\varphi(\mathcal{I})=\mathcal{I}$.  Because
$\varphi$ is fixed in advance, the observed whole and the observed union live
on the same set $\Omega'$ and are compared edge by edge.

\begin{definition}[Emergent and lost edges]\label{def:discrepancy}
The \emph{emergent} edges and the \emph{lost} edges are
\[
E^+_\Phi:=E\bigl(\Phi(s_1\vee s_2)\bigr)\setminus E\bigl(\Phi(s_1)\vee\Phi(s_2)\bigr),
\qquad
E^-_\Phi:=E\bigl(\Phi(s_1)\vee\Phi(s_2)\bigr)\setminus E\bigl(\Phi(s_1\vee s_2)\bigr),
\]
and the \emph{total discrepancy} is $|E^+_\Phi|+|E^-_\Phi|$.
\end{definition}

Emergent edges are structure that the observation of the whole reveals beyond
the observed parts joined by their interaction; lost edges are structure that
the observed parts and their interaction show and the observed whole erases.
In Figure~\ref{fig:square}, $s_1=(a\to b)$, $s_2=(c\to d)$,
$\mathcal{I}=\{(b,c)\}$, and $\Phi=G^{\le2}$.  Each part is a single edge, which
the graph power leaves unchanged, so the observed union is the chain
$a\to b\to c\to d$.  Observing the whole adds $(a,c)$ and $(b,d)$, each joining
the ends of a path of length two through the interaction edge $(b,c)$.  Hence
$E^+_\Phi=\{(a,c),(b,d)\}$ and $E^-_\Phi=\varnothing$.

\begin{figure}[htbp]
\centering
\begin{tikzpicture}[x=1mm,y=1mm,
    frm/.style = {rounded corners=4pt,draw=black!20,fill=black!2,thick},
    op/.style  = {-{Stealth[length=2.6mm,width=2.1mm]},line width=1.1pt,
                  draw=black!60,rounded corners=3pt,shorten >=0.8pt},
    opl/.style = {font=\footnotesize,inner sep=2.5pt},
    obj/.style = {font=\small,inner sep=2pt},
  ]
  \newcommand{\sqpair}[3]{
    \node[v1]    (#1-a) at (#2-5.5,#3+5) {$a$};
    \node[v1,vb] (#1-b) at (#2+5.5,#3+5) {$b$};
    \node[v2,vb] (#1-c) at (#2-5.5,#3-5) {$c$};
    \node[v2]    (#1-d) at (#2+5.5,#3-5) {$d$};
    \draw[eobs] (#1-a) -- (#1-b);
    \draw[eobs] (#1-c) -- (#1-d);
    \begin{scope}[on background layer]
      \node[frm,fit=(#1-a)(#1-d),inner sep=1.6mm] (#1) {};
    \end{scope}}
  \newcommand{\sqchain}[3]{
    \node[v1]    (#1-a) at (#2-17.5,#3) {$a$};
    \node[v1,vb] (#1-b) at (#2-6.5,#3)  {$b$};
    \node[v2,vb] (#1-c) at (#2+6.5,#3)  {$c$};
    \node[v2]    (#1-d) at (#2+17.5,#3) {$d$};
    \draw[eobs] (#1-a) -- (#1-b);
    \draw[ei]   (#1-b) -- (#1-c);
    \draw[eobs] (#1-c) -- (#1-d);
    \begin{scope}[on background layer]
      \node[frm,fit=(#1-a)(#1-d),inner xsep=1.6mm,inner ysep=6.6mm] (#1) {};
    \end{scope}}
  \sqpair{TL}{0}{-17}                
  \sqchain{TR}{50}{0}                
  \sqpair{BL}{50}{-34}               
  \sqchain{UP}{112}{0}               
  \draw[eg] (UP-a) to[bend left=34]  (UP-c);
  \draw[eg] (UP-b) to[bend right=34] (UP-d);
  \sqchain{LO}{112}{-34}             
  \node[obj,anchor=west,align=left] at ($(TL.east)+(1mm,0)$)
       {$(s_1,s_2)$\\[-1pt]\textnormal{\scriptsize\color{obscol}the parts}};
  \node[obj,anchor=south] at (TR.north)
       {$s_1\vee s_2$\ \ \textnormal{\scriptsize\color{obscol}the whole}};
  \node[obj,anchor=south] at (UP.north)
       {$\Phi(s_1\vee s_2)$\ \ \textnormal{\scriptsize\color{obscol}observed whole}};
  \node[obj,anchor=north] at (BL.south)
       {$(\Phi(s_1),\Phi(s_2))$\ \ \textnormal{\scriptsize\color{obscol}observed parts}};
  \node[obj,anchor=north] at (LO.south)
       {$\Phi(s_1)\vee\Phi(s_2)$\ \ \textnormal{\scriptsize\color{obscol}observed union}};
  \draw[op] (TL.north) |- node[opl,above,pos=0.72]{join with {\color{icol}$\mathcal{I}$}} (TR.west);
  \draw[op] (TR) -- node[opl,above]{observe} (UP);
  \draw[op] (TL.south) |- node[opl,below,pos=0.72]{observe each part} (BL.west);
  \draw[op] (BL) -- node[opl,below]{join with {\color{icol}$\mathcal{I}$}} (LO);
  \node[font=\small,text=gcol,align=center] at (112,-17)
       {gap: \ $E^+_\Phi=\{(a,c),(b,d)\}$};
\end{tikzpicture}
\caption{\textbf{Emergence is the gap between joining first and observing
first (Definitions~\ref{def:observed_union} and~\ref{def:discrepancy}).}
The upper way joins the parts with the interaction and then observes the
whole; the lower way observes each part and then joins the observed parts
with the same interaction.  Here $s_1=(a\to b)$, $s_2=(c\to d)$,
$\mathcal{I}=\{(b,c)\}$, and $\Phi=G^{\le2}$.  The green edges, each produced
by a path through $(b,c)$, appear only in the observed whole and form
$E^+_\Phi$; edges present only in the observed union would form $E^-_\Phi$,
which is empty here.}
\alttext{A diagram of two ways from a common start.  On the left, in a light
frame, are the two parts: a to b in blue above c to d in orange.  The upper
way first joins them into the chain a, b, c, d with a dashed red arrow from b
to c, then observes the chain, which gains two green arrows, a to c and b to
d.  The lower way first observes each part, then joins them, giving the chain
a, b, c, d with no green arrows.  The two end results are stacked on the
right, and a green label between them names the two green arrows as the gap.}
\label{fig:square}
\end{figure}

Forming the interaction after the observation fixes the baseline.  For the
identity observation, both ways end at $s_1\vee s_2$, so
there are no emergent and no lost edges: an observation that neither creates
nor discards structure reports no emergence.  For the same reason, the
image $\varphi(\mathcal{I})$ of the interaction already lies in the observed
union, so emergent edges are structure that the observation builds from the
coupling, not the coupling itself.  Lost edges are easy to describe for
monotone observations.

\begin{proposition}[Lost edges are destroyed interaction edges]\label{prop:lost}
If $\Phi$ is monotone, the lost edges are exactly the images of interaction
edges that the observation of the whole destroys:
\[
E^-_\Phi=\varphi(\mathcal{I})\setminus E\bigl(\Phi(s_1\vee s_2)\bigr).
\]
\end{proposition}

\begin{proof}
Each part is a subgraph of the join, so monotonicity gives
$E(\Phi(s_i))\subseteq E(\Phi(s_1\vee s_2))$ for $i=1,2$.  In the difference
that defines $E^-_\Phi$, the edges of $\Phi(s_1)$ and $\Phi(s_2)$ are therefore
removed, and only $\varphi(\mathcal{I})\setminus E(\Phi(s_1\vee s_2))$ remains.
\end{proof}

For monotone observations, lost edges are thus found among the images of the
interaction edges
(thresholding away a light interaction edge loses exactly that edge); emergent
edges are what Section~\ref{sec:main} bounds.  Both directions already appear
for simple deletion rules.

\begin{example}[Deleting by degree]\label{ex:degree}
Two observations delete edges by the total degrees of their endpoints, both
with $\varphi$ the identity: $\Phi^{\mathrm{hi}}$ deletes every edge whose
endpoints both have degree at least $2$, and $\Phi^{\mathrm{lo}}$
deletes every edge whose endpoints both have degree less than $2$.  Joining
raises degrees, and the two observations respond oppositely.  Let
$s_1=(a\to b)$ and $s_2=(c\to d)$.

\emph{Deleting between high-degree vertices loses edges.}  Take
$\mathcal{I}=\{(a,c),(b,d)\}$.  Every vertex of a part has degree $1$, so
$\Phi^{\mathrm{hi}}(s_i)=s_i$ and the observed union has the four edges
$(a,b)$, $(c,d)$, $(a,c)$, and $(b,d)$.  In the join, every vertex has degree
$2$ and every edge is deleted.  Hence $E^+_{\Phi^{\mathrm{hi}}}=\varnothing$, and
$E^-_{\Phi^{\mathrm{hi}}}$ consists of all four edges: the observed whole
erases every edge that the observed parts and their interaction show.

\emph{Deleting between low-degree vertices gains edges.}  Take
$\mathcal{I}=\{(b,c)\}$.  In each part both endpoints have degree $1$, so
$\Phi^{\mathrm{lo}}(s_i)$ has no edges and the observed union has the single
edge $(b,c)$.  In the join, $b$ and $c$ have degree $2$, so all three edges
survive, $E^+_{\Phi^{\mathrm{lo}}}=\{(a,b),(c,d)\}$, and
$E^-_{\Phi^{\mathrm{lo}}}=\varnothing$: the interaction
supplies the degree that each part lacks.

Both observations read all the edges at a vertex together: under
$\Phi^{\mathrm{lo}}$, two edges entering one vertex, which lie on no common
route, survive together although each alone is deleted.  Section~\ref{sec:main} turns to observations that read
one route at a time, for which every emergent edge can be traced to a route
across the interface.
\end{example}

\section{Emergence is bounded by boundary-crossing paths}\label{sec:main}

Joining two parts adds exactly the routes that cross the interface; for
observations that build each observed edge from a short route, this yields
an identity, a bound, and an attribution for emergence.

\subsection{Observations acting on paths}

We ask that an observation be a rule applied to one short route at a time:
whether a route yields an observed edge, and which edge, is decided by the
route alone.  Paths, cycles, and $\mathcal P_L(G)$ are as in
Section~\ref{sec:paths}.

\begin{definition}[Observations acting on paths]\label{def:path_obs}
An observation $\Phi$ \emph{acts on paths of length at most $L$} if each path $p$ of
length at most $L$ either \emph{produces} a single edge, written $\prodedge{p}$, or
produces nothing, in a way that depends on $p$ alone and not on the graph that
contains it, and if for every graph $G$ the edges of $\Phi(G)$ are exactly the edges
produced by the paths of $G$.  A path of length one that produces an edge produces
its own image, $\prodedge{(u,v)}=(\varphi(u),\varphi(v))$.
\end{definition}

For a set $\mathcal Q$ of paths, $\prodedge{\mathcal Q}$ is the set of edges
its members produce, so $E(\Phi(G))=\prodedge{\mathcal P_L(G)}$.  The
observations of Section~\ref{sec:framework} act on paths by these rules:
\begin{itemize}[leftmargin=1.2em,itemsep=0pt,parsep=0pt,topsep=2pt]
\item \emph{Identity} and \emph{thresholding} at $\tau$ ($L=1$): an edge (of
  weight at least $\tau$) produces itself.
\item \emph{Coarse-graining} ($L=1$): an edge produces the pair of blocks of its
  endpoints, if they differ.
\item \emph{Graph power} $G^{\le L}$~\cite{bang2009digraphs}: a path
  produces the edge from its first to its last vertex.
\item \emph{Thresholded power}: the same, for a path whose edges all
  weigh at least $\tau$.
\item \emph{Coarse-grained power}: a path produces the pair of blocks of
  its first and last vertices, if they differ.
\item \emph{Drop sinks} ($L=2$): a path of length two, or a two-cycle, produces
  its first edge.
\item \emph{Edges on a cycle} of length at most $L$: each rotation of a cycle
  produces its first edge.
\item \emph{Continuation filter}: a path of length $L$ produces its first
  edge (every longer path begins with one of length $L$ that has the same
  first edge).
\end{itemize}
Such observations are monotone, since the paths of a subgraph are paths of
the graph that contains it.  By Proposition~\ref{prop:lost}, their lost edges are
therefore the images of interaction edges that the observation destroys.

\subsection{Boundary-crossing paths and the main theorem}

\begin{definition}[Boundary-crossing paths]\label{def:crossing}
A path of $s_1\vee s_2$ is \emph{boundary-crossing} if it traverses an
interaction edge.  We write $\mathcal P_\partial$ for the boundary-crossing members
of $\mathcal P_L(s_1\vee s_2)$ (with cycles when $\Phi$ reads them),
$\mathcal S_\partial\subseteq\mathcal P_\partial$ for those that
produce an edge, and $N_\partial:=|\mathcal S_\partial|$.
\end{definition}

\begin{lemma}[New paths cross the boundary]\label{lem:crossing}
If $V_1\cap V_2=\varnothing$, the paths of $s_1\vee s_2$ of length at most $L$
fall into three disjoint classes, the paths of $s_1$, the paths of $s_2$, and
the boundary-crossing paths:
\[
\mathcal P_L(s_1\vee s_2)=\mathcal P_L(s_1)\sqcup\mathcal P_L(s_2)\sqcup\mathcal P_\partial .
\]
\end{lemma}

\begin{proof}
Each part is a subgraph of the join, an interaction edge belongs to neither
part, and no path lies in both parts, since $V_1\cap V_2=\varnothing$; so the
three classes are disjoint subsets of $\mathcal P_L(s_1\vee s_2)$.  A path (or
cycle) that avoids $\mathcal I$ uses only edges of $s_1$ and $s_2$; consecutive
edges share a vertex, and every edge of $s_i$ has both endpoints in $V_i$, so
all its edges lie in one part, and it is a path of that part
(Figure~\ref{fig:crossing}(b)).
\end{proof}

A boundary-crossing path or cycle, split at an interaction edge $(u,v)$ that
it traverses, consists of a path of length $i$ ending at $u$, the edge
$(u,v)$, and a path of length $j$ starting at $v$, with $i+j\le L-1$ (here a
path of length $0$ is a single vertex); the two pieces determine it.

Routes inside a part are already routes of that part on its own, so the
edges they produce appear in the observed union.  Anything new comes from a
route that no part contains, which by Lemma~\ref{lem:crossing} crosses the
interface: an emergent edge may join vertices far from the interface, but
its cause is always a route through it.

\begin{theorem}[Emergence is produced by boundary-crossing paths]\label{thm:main}
Let $s_1,s_2$ have disjoint vertex sets, let $\mathcal I$ be an interaction, and let
$\Phi$ act on paths of length at most $L$.  Write $U:=E\bigl(\Phi(s_1)\vee\Phi(s_2)\bigr)$.
\begin{enumerate}[label=(\roman*)]
\item $E\bigl(\Phi(s_1\vee s_2)\bigr)=E\bigl(\Phi(s_1)\bigr)\cup
E\bigl(\Phi(s_2)\bigr)\cup\prodedge{\mathcal S_\partial}$, and hence
$E^+_\Phi=\prodedge{\mathcal S_\partial}\setminus U$: the emergent edges are
exactly the edges that boundary-crossing paths produce and that the observed
union does not already contain.
\item $|E^+_\Phi|\le N_\partial$.  More precisely, choosing for each emergent edge one
boundary-crossing path that produces it defines an injective map
$w:E^+_\Phi\to\mathcal S_\partial$, a \emph{witness map}; $w(e)$ is a
\emph{witness} of $e$.
\item $|E^+_\Phi|=N_\partial$ if and only if distinct paths in $\mathcal S_\partial$
produce distinct edges and none of these edges lies in $U$.
\end{enumerate}
\end{theorem}

\begin{proof}
(i)~Since $\Phi$ acts on paths, $E(\Phi(G))=\prodedge{\mathcal P_L(G)}$ for
every graph $G$, where the edge that a path produces depends on the path
alone.  By Lemma~\ref{lem:crossing}, $\mathcal P_L(s_1\vee s_2)$ is the union
of $\mathcal P_L(s_1)$, $\mathcal P_L(s_2)$, and $\mathcal P_\partial$, whose
members that produce an edge form $\mathcal S_\partial$; this gives the first
identity.  The edges of $\Phi(s_1)$ and $\Phi(s_2)$ lie in $U$, which gives the
second.
(ii)~By (i), each $e\in E^+_\Phi$ is produced by some path
$w(e)\in\mathcal S_\partial$.  A path produces at most one edge, so
$\prodedge{w(e)}=e$; hence $w$ is injective and $|E^+_\Phi|\le N_\partial$.
(iii)~By (i), $|E^+_\Phi|=|\prodedge{\mathcal S_\partial}\setminus
U|\le|\prodedge{\mathcal S_\partial}|\le|\mathcal S_\partial|$, and the
conditions in (iii) are exactly those for equality in the two steps.
\end{proof}

Without interaction edges there are no boundary-crossing paths, and hence no
emergence.  By Theorem~\ref{thm:main}(ii), distinct emergent edges need distinct
boundary-crossing paths, and the observation decides which of these paths
produce an edge.  A single interaction edge can carry many of them
(Example~\ref{ex:bowtie}), much as the few edges
between communities carry many shortest paths~\cite{girvan2002community},
so the bound measures the routes that the interface opens, not its size.

The slack of the bound follows from Theorem~\ref{thm:main}(i): it is exactly
$N_\partial-|E^+_\Phi|=\bigl(|\mathcal
S_\partial|-|\prodedge{\mathcal S_\partial}|\bigr)+|\prodedge{\mathcal
S_\partial}\cap U|$; it comes from boundary-crossing paths that produce the
same edge and from produced edges that the observed union already contains,
such as images of interaction edges.

\begin{corollary}[Sharpened bound]\label{cor:sharp}
In the setting of Theorem~\ref{thm:main}, let $\mathcal S^{\ge2}_\partial$
be the paths in $\mathcal S_\partial$ of length at least two.  Then
$E^+_\Phi=\prodedge{\mathcal S^{\ge2}_\partial}\setminus U$ and
$|E^+_\Phi|\le|\mathcal S^{\ge2}_\partial|$, with equality if and only if
distinct paths in $\mathcal S^{\ge2}_\partial$ produce distinct edges and none
of these edges lies in $U$.
\end{corollary}

\begin{proof}
A boundary-crossing path of length one is an interaction edge $(u,v)$, and
an edge it produces is $(\varphi(u),\varphi(v))\in\varphi(\mathcal
I)\subseteq U$, since observed graphs are loopless.  Such paths add nothing
to $\prodedge{\mathcal S_\partial}\setminus U$, and the arguments of
Theorem~\ref{thm:main}(ii) and~(iii) apply to the remaining paths.
\end{proof}

\begin{corollary}[The total discrepancy]\label{cor:discrepancy}
In the setting of Theorem~\ref{thm:main},
$|E^+_\Phi|+|E^-_\Phi|\le|\mathcal P_\partial|$, the number of all
boundary-crossing paths, producing or not.
\end{corollary}

\begin{proof}
The witness map sends $E^+_\Phi$ injectively into $\mathcal S_\partial$.
By Proposition~\ref{prop:lost}, each lost edge is the image of an
interaction edge that produces nothing (otherwise its image would lie in
$E(\Phi(s_1\vee s_2))$); choosing one such interaction edge per lost edge maps
$E^-_\Phi$ injectively into $\mathcal P_\partial\setminus\mathcal S_\partial$.
\end{proof}

\begin{corollary}[Locality and cost]\label{cor:local}
In the setting of Theorem~\ref{thm:main}:
\begin{enumerate}[label=(\roman*)]
\item every emergent edge is produced by a path of length at most $L$
  through an interaction edge; the vertices of that path before the
  interaction edge reach $\partial$, and those after it are reached from
  $\partial$, within $L-1$ steps;
\item $\mathcal S_\partial$, and with it $N_\partial$ and the set
  $\prodedge{\mathcal S_\partial}$ that contains every emergent edge, is
  found by a backward search from the tail and a forward search from the
  head of each interaction edge, of combined depth $L-1$, testing each path
  or cycle found for production.  With $d\ge2$ the maximum in- or out-degree
  of $s_1\vee s_2$, the search examines $O(L\,|\mathcal I|\,d^{\,L-1})$
  pairs of pieces, a number independent of $|V_1|$ and $|V_2|$, without
  applying $\Phi$ to $s_1$, $s_2$, or their join (Appendix~\ref{app:computation}
  gives the details, including what building a witness map requires).
\end{enumerate}
\end{corollary}

\begin{proof}
By Theorem~\ref{thm:main}(i), every emergent edge is produced by a path of
$\mathcal S_\partial$; split at an interaction edge, whose endpoints lie in
$\partial$, this path has two pieces of length at most $L-1$, which
gives~(i).  The searches in~(ii) list the pieces: with $i$ steps back and $j$
forward there are at most $d^{\,i}d^{\,j}$ pairs, each forming at most one path
or cycle, so the search examines at most
$\sum_{t=0}^{L-1}(t+1)d^{\,t}\le 2L\,d^{\,L-1}$ pairs per interaction edge for
$d\ge2$.
\end{proof}

When the interface is dense or its vertices have high degree, so that
$|\mathcal I|\,d^{\,L-1}$ exceeds the number of edges of $s_1\vee s_2$, the
walks of length at most $L$ that traverse an interaction edge give a bound
that is cheaper to compute, at a cost that grows only linearly in $L$
(Appendix~\ref{app:computation}).

The witness map attributes emergence route by route: it names the vertices
and edges on each witness, in $s_1$, $s_2$, and $\mathcal I$, as responsible
for its emergent edge.  When several paths produce the same emergent edge,
splitting the edge equally among them, as betweenness splits each pair of
vertices among its shortest paths~\cite{freeman1977set}, gives an attribution
that involves no choice of witness.

\begin{definition}[Shares and credit]\label{def:shares}
In the setting of Theorem~\ref{thm:main}, for $e\in E^+_\Phi$ let
$\mathcal S_\partial(e):=\{p\in\mathcal S_\partial:\prodedge{p}=e\}$, which is
nonempty by Theorem~\ref{thm:main}(i).  The \emph{share} of a path
$p\in\mathcal S_\partial(e)$ is $1/|\mathcal S_\partial(e)|$, and a path of
$\mathcal S_\partial$ that produces no emergent edge has share $0$.  The
\emph{credit} of a vertex, an edge, or a set of edges of $s_1\vee s_2$, such as
the interaction edges from one part to another, is the total share of the
paths of $\mathcal S_\partial$ that pass through it: that visit the vertex, or
traverse the edge or one of the edges of the set.
\end{definition}

The shares sum to $|E^+_\Phi|$, which is also the credit of $\mathcal I$,
since every boundary-crossing path traverses an interaction edge.  The share
of a path is the probability that a witness map drawn uniformly at random uses
it, so the credit of a vertex or an edge is the expected number of witnesses
that pass through it.  Shares also tell which emergent edges survive the
removal of interaction edges, and credits bound how many are lost
(Proposition~\ref{prop:intervene}).

\begin{remark}[Many parts]\label{rem:multipart}
For parts $s_1,\dots,s_k$ with pairwise disjoint vertex sets, an
interaction $\mathcal I\subseteq\bigcup_{i\ne j}V_i\times V_j$, and the
observed union formed by the graphs $\Phi(s_i)$ together with
$\varphi(\mathcal I)$, a path that avoids $\mathcal I$ stays inside one part.
Hence Lemma~\ref{lem:crossing},
Theorem~\ref{thm:main}, Corollaries~\ref{cor:sharp}--\ref{cor:local}, the disjoint case of
Theorem~\ref{thm:scalar} (with the scalar discrepancy of
Definition~\ref{def:discrepancy_scalar} taken as
$\Phi(s_1\vee\dots\vee s_k)-\sum_i\Phi(s_i)$), and
Propositions~\ref{prop:lost}, \ref{prop:cut}, and~\ref{prop:intervene} hold
verbatim, with the same proofs, and Definition~\ref{def:shares} applies
unchanged (Figure~\ref{fig:overview}; checked in Sections~\ref{sec:validity}
and~\ref{sec:real}).
\end{remark}

\subsection{A single bridge: emergence multiplies reach}\label{sec:bridge}

For graph powers and an interaction that runs one way, the boundary-crossing
paths and the emergent edges have closed forms in counts of paths and of
reached vertices in the two parts.  Write
$\mathrm{dist}_s(x,y)$ for the length of a shortest path from $x$ to $y$ in
$s$ ($0$ if $x=y$, and $\infty$ if there is none).  For $a\in V_1$ and
$b\in V_2$, let $\pi^{\mathrm{in}}_i(a)$ be the number of paths of length $i$
in $s_1$ that end at $a$ and $\pi^{\mathrm{out}}_j(b)$ the number of paths of
length $j$ in $s_2$ that start at $b$, where, for these counts, a path of
length $0$ is a single vertex, so that
$\pi^{\mathrm{in}}_0(a)=\pi^{\mathrm{out}}_0(b)=1$; and let
$\rho^{\mathrm{in}}_i(a)$ be the number of vertices $x$ of $s_1$ with
$\mathrm{dist}_{s_1}(x,a)=i$ and $\rho^{\mathrm{out}}_j(b)$ the number of
vertices $y$ of $s_2$ with $\mathrm{dist}_{s_2}(b,y)=j$.  The vertices of $s_1$
with a path to $a$ form the \emph{in-reach} of $a$, and those of $s_2$ reached
from $b$ the \emph{out-reach} of $b$.

\begin{proposition}[Bridge law]\label{prop:bridge}
Let $\Phi=G^{\le L}$ and $\mathcal I\subseteq V_1\times V_2$.
\begin{enumerate}[label=(\roman*)]
\item The boundary-crossing paths are exactly the paths formed by a path of
length $i$ in $s_1$ ending at the tail $a$ of an interaction edge $(a,b)$, that
edge, and a path of length $j$ in $s_2$ starting at $b$, with $i+j\le L-1$.
Hence
\[
N_\partial=\sum_{(a,b)\in\mathcal I}\ \sum_{i+j\le L-1}\pi^{\mathrm{in}}_i(a)\,\pi^{\mathrm{out}}_j(b).
\]
\item $E^-_\Phi=\varnothing$, and the emergent edges are the cross pairs within
combined distance $L-1$ of an interaction edge:
\[
E^+_\Phi=\Bigl\{(x,y)\in V_1\times V_2:\ \min_{(a,b)\in\mathcal I}\bigl[\mathrm{dist}_{s_1}(x,a)+\mathrm{dist}_{s_2}(b,y)\bigr]\le L-1\Bigr\}\setminus\mathcal I .
\]
For a single interaction edge $(a,b)$,
$|E^+_\Phi|=\sum_{i+j\le L-1}\rho^{\mathrm{in}}_i(a)\,\rho^{\mathrm{out}}_j(b)-1$.
\item $|E^+_\Phi|\le N_\partial-|\mathcal I|$, the sharpened bound of
Corollary~\ref{cor:sharp}, with equality if and only if every pair in
$V_1\times V_2$ is joined in $s_1\vee s_2$ by at most one path of length at
most $L$.  For a single interaction edge $(a,b)$, equality holds if and only
if every vertex of $s_1$ has at most one path of length at most $L-1$ to $a$,
and $b$ has at most one path of length at most $L-1$ to every vertex of $s_2$.
\end{enumerate}
\end{proposition}

\begin{proof}
(i)~No edge runs from $V_2$ to $V_1$, so a path that enters $V_2$ stays there,
and a boundary-crossing path traverses exactly one interaction edge $(a,b)$.
Split at that edge, it is a path of $s_1$ ending at $a$, the edge, and a path
of $s_2$ starting at $b$, of lengths $i$ and $j$ with $i+j\le L-1$.  Conversely, since
$V_1\cap V_2=\varnothing$, joining such paths by $(a,b)$ repeats no vertex.
Under $G^{\le L}$ every path produces the edge between its ends, which differ,
so $N_\partial$ is the number of these paths.
(ii)~Every interaction edge produces itself, so $E^-_\Phi=\varnothing$ by
Proposition~\ref{prop:lost}.  By~(i), the boundary-crossing paths produce
exactly the pairs $(x,y)\in V_1\times V_2$ with
$\mathrm{dist}_{s_1}(x,a)+\mathrm{dist}_{s_2}(b,y)\le L-1$ for some
$(a,b)\in\mathcal I$; the only such pairs in $U$ are the edges of $\mathcal I$,
since $E(\Phi(s_i))\subseteq V_i\times V_i$, and Theorem~\ref{thm:main}(i) gives
$E^+_\Phi$.  For one edge $(a,b)$, the pairs with
$\mathrm{dist}_{s_1}(x,a)=i$ and $\mathrm{dist}_{s_2}(b,y)=j$ number
$\rho^{\mathrm{in}}_i(a)\,\rho^{\mathrm{out}}_j(b)$, and $(a,b)$ is the pair with
$i=j=0$.
(iii)~Let $c(x,y)$ be the number of paths of length at most $L$ from $x\in V_1$
to $y\in V_2$.  By~(i) these are the boundary-crossing paths, so
$N_\partial=\sum_{x,y}c(x,y)$, and by~(ii) $c(x,y)\ge1$ exactly on the disjoint
union $E^+_\Phi\cup\mathcal I$.  Hence $N_\partial\ge|E^+_\Phi|+|\mathcal I|$,
with equality exactly when $c\le1$; and $N_\partial-|\mathcal I|$ is the
sharpened bound, because the paths of length one in $\mathcal S_\partial$ are
the interaction edges.  For one edge $(a,b)$, a path from $x$ to $y$ consists
of a path of $s_1$ from $x$ to $a$, the edge, and a path of $s_2$ from $b$ to
$y$, each of length at most $L-1$.  Thus $c(x,b)$ counts the paths of length at
most $L-1$ from $x$ to $a$ in $s_1$, $c(a,y)$ those from $b$ to $y$ in $s_2$,
and $c(x,y)\le c(x,b)\,c(a,y)$; so $c\le1$ exactly when the stated condition
holds.
\end{proof}

One interaction edge couples every route into its tail with every route out
of its head, of combined length at most $L-1$, so the boundary-crossing paths
are counted by sums of products of the numbers of paths into the tail and out
of the head, and the emergent edges by sums of products of the numbers of
vertices at each distance from them: part~(i) counts what the search of
Corollary~\ref{cor:local}(ii) finds, and part~(ii) makes the locality of
Corollary~\ref{cor:local}(i) exact.  The boundary-crossing paths are routes and
the emergent edges are pairs of endpoints; beyond the interaction itself,
the two counts coincide exactly when routes are unique, as in trees.

\begin{example}[Bow-tie]\label{ex:bowtie}
Let $s_1$ be an in-fan $u_1,\dots,u_k\to a$ and $s_2$ an out-fan
$v\to w_1,\dots,w_m$, joined by the single interaction edge $(a,v)$, and let
$\Phi=G^{\le3}$ (Figure~\ref{fig:bowtie}).  Every vertex of $s_1$ reaches $a$,
and every vertex of $s_2$ is reached from $v$, by exactly one path, of length at
most one.  By Proposition~\ref{prop:bridge}, the boundary-crossing paths are
$(a,v)$, $(u_i,a,v)$, $(a,v,w_j)$, and $(u_i,a,v,w_j)$, so
$N_\partial=1+k+m+km=(1+k)(1+m)$, and equality holds in~(iii):
$|E^+_\Phi|=N_\partial-1=k+m+km$, the edges $(u_i,v)$, $(a,w_j)$, and
$(u_i,w_j)$; the bound exceeds the emergence only by the interaction edge.  The term $km$, which pairs each vertex that reaches the
interface with each vertex that the interface reaches, is a multiplicative
signature of emergence under graph powers: a single bridge between internally
structured parts can yield emergence that grows as the product of the in-reach
of its tail and the out-reach of its head.
\end{example}

\begin{figure}[htbp]
\centering
\def\btRow#1#2{%
  \foreach \vl/\vx/\vp [count=\vi] in {#2} {%
    \node[bc\vp] (r\vi) at (\vx,#1) {$\vl$};
    \ifnum\vi>1
      \pgfmathtruncatemacro\vj{\vi-1}%
      \ifnum\vp=\btLast\relax \draw[bce\vp] (r\vj) -- (r\vi);
      \else \draw[bcei] (r\vj) -- (r\vi); \fi
    \fi
    \xdef\btLast{\vp}%
  }}
\begin{tikzpicture}[x=1mm,y=1mm,
  bc1/.style={circle,draw=s1col,fill=s1col!16,line width=.6pt,minimum size=4.3mm,inner sep=0pt,font=\scriptsize},
  bc2/.style={circle,draw=s2col,fill=s2col!18,line width=.6pt,minimum size=4.3mm,inner sep=0pt,font=\scriptsize},
  bce1/.style={-{Stealth[length=1.5mm]},line width=.7pt,draw=s1col!85!black},
  bce2/.style={-{Stealth[length=1.5mm]},line width=.7pt,draw=s2col!85!black},
  bcei/.style={-{Stealth[length=1.5mm]},line width=.9pt,draw=icol,densely dashed},
  hdr/.style={font=\scriptsize,anchor=base,text=black!75},
  chip/.style={rounded corners=2pt,inner sep=1.6pt,font=\scriptsize,minimum height=4.3mm},
]
\node[panel,anchor=west] at (-7,17) {(a)};
\node[font=\small,anchor=west] at (-0.8,17) {parts and interaction};
\node[v1]    (u1) at (0,7)   {$u_1$};
\node[v1]    (u2) at (0,-7)  {$u_2$};
\node[v1,vb] (a)  at (12,0)  {$a$};
\node[v2,vb] (v)  at (25,0)  {$v$};
\node[v2]    (w1) at (37,7)  {$w_1$};
\node[v2]    (w2) at (37,-7) {$w_2$};
\begin{scope}[on background layer]
  \node[blob1,fit=(u1)(u2)(a)] (S1) {};
  \node[blob2,fit=(v)(w1)(w2)] (S2) {};
\end{scope}
\draw[e,s1col] (u1)--(a); \draw[e,s1col] (u2)--(a);
\draw[e,s2col] (v)--(w1); \draw[e,s2col] (v)--(w2);
\draw[ei] (a)--(v);
\node[lbl,icol,anchor=south] at (18.5,0.6) {$\mathcal I$};
\node[lbl,s1col,anchor=north,align=center] at (S1.south) {$s_1$: in-fan\\[-1pt]$k=2$};
\node[lbl,s2col,anchor=north,align=center] at (S2.south) {$s_2$: out-fan\\[-1pt]$m=2$};
\begin{scope}[shift={(50,0)}]
\node[panel,anchor=west] at (-4,17) {(b)};
\node[font=\small,anchor=west] at (2.2,17) {boundary-crossing paths under $\Phi=G^{\le 3}$};
\def\xu{0}\def\xa{9}\def\xv{20}\def\xw{29}
\draw[icol!55,line width=.6pt,dash pattern=on 1.2pt off 1.2pt] (14.5,11) -- (14.5,-14.2);
\node[hdr] at (14.5,11.4) {interface};
\node[hdr] at (39,11.4) {count};
\node[hdr] at (56,11.4) {produces};
\btRow{7}{a/\xa/1,v/\xv/2}
\btRow{1}{u_i/\xu/1,a/\xa/1,v/\xv/2}
\btRow{-5}{a/\xa/1,v/\xv/2,w_j/\xw/2}
\btRow{-11}{u_i/\xu/1,a/\xa/1,v/\xv/2,w_j/\xw/2}
\foreach \y/\cnt in {7/1, 1/k, -5/m, -11/km} \node[font=\scriptsize] at (39,\y) {$\cnt$};
\foreach \y in {7,1,-5,-11} \draw[-{Stealth[length=1.4mm]},black!55,line width=.5pt] (44,\y) -- (48,\y);
\node[chip,fill=obscol!15,text=black!70]      at (56,7)   {$(a,v)$};
\node[chip,fill=gcol!20,text=gcol!55!black]   at (56,1)   {$(u_i,v)$};
\node[chip,fill=gcol!20,text=gcol!55!black]   at (56,-5)  {$(a,w_j)$};
\node[chip,fill=gcol!20,text=gcol!55!black]   at (56,-11) {$(u_i,w_j)$};
\node[font=\scriptsize,anchor=west,text=black!65] at (63,7) {the interaction itself};
\node[font=\scriptsize,anchor=west,gcol!70!black] at (63,1)   {in-reach};
\node[font=\scriptsize,anchor=west,gcol!70!black] at (63,-5)  {out-reach};
\node[font=\scriptsize,anchor=west,gcol!70!black] at (63,-11) {in-reach $\times$ out-reach};
\draw[black!45,line width=.5pt] (-3,-14.8) -- (92,-14.8);
\node[font=\scriptsize,anchor=north] at (45,-15.6) {%
  \begin{tabular}{@{}r@{\;}l@{\qquad}l@{}}
  $N_\partial$ & $=1+k+m+km=9$ & boundary-crossing paths\\[1pt]
  \color{gcol!55!black}$|E^+_\Phi|$ & \color{gcol!55!black}$=k+m+km=8$ & \color{gcol!55!black}emergent edges
  \end{tabular}};
\end{scope}
\end{tikzpicture}
\caption{\textbf{Emergence multiplies reach on the two sides of the
interface (Example~\ref{ex:bowtie}, $k=m=2$).}
\textbf{(a)}~An in-fan $s_1$ and an out-fan $s_2$ joined by the interaction
edge $(a,v)$.
\textbf{(b)}~Under $\Phi=G^{\le3}$, the boundary-crossing paths come in four
kinds, one row each: the interaction edge, the $k$ paths that end with it,
the $m$ that start with it, and the $km$ that pass through it.  Each kind
produces its own kind of edge.  The first produces the interaction edge,
which the observed union already contains; the other three produce the
emergent edges.}
\alttext{Left: a bow-tie graph in two shaded regions. In the blue region,
the in-fan s1 with k equal to 2, two vertices u1 and u2 point to a
double-ringed vertex a. A dashed red interaction edge runs from a to a
double-ringed vertex v in the orange region, the out-fan s2 with m equal to
2, where v points to two vertices w1 and w2. Right: four rows, each drawing
one kind of path that crosses a vertical dotted interface line: a to v
(count 1), u to a to v (count k), a to v to w (count m), and u to a to v to
w (count k times m). An arrow leads from each row to the edge it produces:
a to v, labeled the interaction itself, and the green emergent edges u to v,
a to w, and u to w, labeled in-reach, out-reach, and in-reach times
out-reach. The totals read N equals 1 plus k plus m plus k times m equals 9
boundary-crossing paths, and k plus m plus k times m equals 8 emergent
edges.}
\label{fig:bowtie}
\end{figure}

For random parts, Proposition~\ref{prop:bridge} gives the exact expected
numbers of boundary-crossing paths and of emergent edges, which
Section~\ref{sec:bridging}
compares with experiment.

\begin{corollary}[Bridging random modules]\label{cor:bridge_random}
Let $s_1$ and $s_2$ be independent random digraphs on $n_1$ and $n_2$ vertices
in which each ordered pair of distinct vertices is an edge independently with
probability $p_1$ and $p_2$, respectively, let $\mathcal I=\{(a,b)\}$ with $a\in V_1$ and
$b\in V_2$ chosen independently of the parts, and let $\Phi=G^{\le L}$.  Then
\[
\mathbb E[N_\partial]=\sum_{i+j\le L-1}(n_1-1)_i\,p_1^{\,i}\,(n_2-1)_j\,p_2^{\,j},
\qquad
\mathbb E\bigl[|E^+_\Phi|\bigr]=\sum_{i+j\le L-1}\mathbb E[\rho^{\mathrm{in}}_i(a)]\,\mathbb E[\rho^{\mathrm{out}}_j(b)]-1,
\]
where $(n)_i=n(n-1)\cdots(n-i+1)$, and $\rho^{\mathrm{in}}_i(a)$ and
$\rho^{\mathrm{out}}_i(b)$ have the law of $F_i$ in the chain $F_0=1$,
$R_0=n-1$, $F_{i+1}\sim\mathrm{Binomial}\bigl(R_i,1-(1-p)^{F_i}\bigr)$,
$R_{i+1}=R_i-F_{i+1}$, with $(n,p)$ those of the part.
\end{corollary}

\begin{proof}
A sequence of $i$ distinct vertices of $s_1$ other than $a$, followed by
$a$, is a path into $a$ with probability $p_1^{\,i}$, and there are $(n_1-1)_i$
such sequences; paths out of $b$ are counted in the same way.  The in-reach of
$a$ depends only on $s_1$ and the out-reach of $b$ only on $s_2$, and the law of
each part is invariant under relabeling, so Proposition~\ref{prop:bridge}(i)
and~(ii) give the two expectations.  In a breadth-first search from $b$, the
edges from the $F_i$ vertices at distance $i$ to the $R_i$ vertices not yet
reached have not been examined, so each of these vertices joins the next layer
independently with probability $1-(1-p)^{F_i}$.  Reversing every edge of $s_1$
preserves its law, so the same chain gives the in-reach of $a$.
\end{proof}

\subsection{Scalar observations}\label{sec:scalar}

\begin{definition}[Scalar discrepancy]\label{def:discrepancy_scalar}
For a real-valued function $\Phi$ of graphs, such as the number of paths of
length at most $L$, the \emph{scalar discrepancy} is
$\Delta_\Phi(s_1,s_2):=\Phi(s_1\vee s_2)-\Phi(s_1)-\Phi(s_2)$.
\end{definition}

Many network measures add up contributions of routes: path counts, and the
attenuated walk sums of Katz~\cite{katz1953new} and of
communicability~\cite{estrada2008communicability}, which count walks rather
than simple paths.  For measures built from simple paths, the
boundary-crossing paths give the exact emergence, not only a bound; for walk
sums, the walks that traverse an interaction edge play the same role
(Appendix~\ref{app:computation}).

\begin{figure}[htbp]
\centering
\def\scRow#1#2{%
  \foreach \vl/\vx/\vp [count=\vi] in {#2} {%
    \node[sc\vp] (r\vi) at (\vx,#1) {$\vl$};
    \ifnum\vi>1
      \pgfmathtruncatemacro\vj{\vi-1}%
      \ifnum\vp=\scLast\relax \draw[sce\vp] (r\vj) -- (r\vi);
      \else \draw[scei] (r\vj) -- (r\vi); \fi
    \fi
    \xdef\scLast{\vp}%
  }}
\begin{tikzpicture}[x=1mm,y=1mm,
  sc1/.style={circle,draw=s1col,fill=s1col!16,line width=.6pt,minimum size=3.9mm,inner sep=0pt,font=\scriptsize},
  sc2/.style={circle,draw=s2col,fill=s2col!18,line width=.6pt,minimum size=3.9mm,inner sep=0pt,font=\scriptsize},
  sce1/.style={-{Stealth[length=1.5mm]},line width=.7pt,draw=s1col!85!black},
  sce2/.style={-{Stealth[length=1.5mm]},line width=.7pt,draw=s2col!85!black},
  scei/.style={-{Stealth[length=1.5mm]},line width=.9pt,draw=icol,densely dashed},
  band/.style={rounded corners=3pt,line width=.5pt},
  coef/.style={rounded corners=2pt,inner sep=2pt,font=\scriptsize,minimum height=4.2mm},
  gname/.style={font=\scriptsize,align=left,anchor=west},
  hdr/.style={font=\scriptsize,align=center,anchor=south},
  tbar/.style={rounded corners=1.5pt,line width=.5pt,minimum height=4.3mm},
]
\node[panel,anchor=west] at (-9,13) {(a)};
\node[font=\small,anchor=west] at (-2.6,13) {two disjoint parts: path count ($q\equiv1$, $L=4$)};
\node[v1] (A) at (0,0) {$a$};
\node[v1] (B) at (11,0) {$b$};
\node[v1,vb] (C) at (22,0) {$c$};
\node[v2,vb] (D) at (35,0) {$d$};
\node[v2] (E) at (46,0) {$e$};
\begin{scope}[on background layer]
  \node[blob1,fit=(A)(C)] (S1) {};
  \node[blob2,fit=(D)(E)] (S2) {};
\end{scope}
\draw[e,s1col] (A)--(B); \draw[e,s1col] (B)--(C);
\draw[e,s2col] (D)--(E);
\draw[ei] (C)--(D);
\node[lbl,s1col,anchor=south] at (S1.north) {$s_1$};
\node[lbl,s2col,anchor=south] at (S2.north) {$s_2$};
\node[lbl,icol,anchor=south] at (28.5,1.2) {$\mathcal I$};
\draw[icol!55,line width=.6pt,dash pattern=on 1.2pt off 1.2pt] (28.5,-5.5) -- (28.5,-67.2);
\node[font=\scriptsize,icol,anchor=north] at (28.5,-67.4) {interface};
\def\bl{-5.5}\def\br{51.5}
\begin{scope}[on background layer]
  \fill[band,s1col!7]  (\bl,-8.7)  rectangle (\br,-24.7);
  \fill[band,s2col!8]  (\bl,-26.7) rectangle (\br,-32.5);
  \fill[band,gcol!8]   (\bl,-34.5) rectangle (\br,-66.5);
\end{scope}
\scRow{-11.5}{a/0/1,b/11/1}
\scRow{-16.7}{b/11/1,c/22/1}
\scRow{-21.9}{a/0/1,b/11/1,c/22/1}
\scRow{-29.6}{d/35/2,e/46/2}
\scRow{-37.5}{c/22/1,d/35/2}
\scRow{-42.7}{b/11/1,c/22/1,d/35/2}
\scRow{-47.9}{c/22/1,d/35/2,e/46/2}
\scRow{-53.1}{a/0/1,b/11/1,c/22/1,d/35/2}
\scRow{-58.3}{b/11/1,c/22/1,d/35/2,e/46/2}
\scRow{-63.5}{a/0/1,b/11/1,c/22/1,d/35/2,e/46/2}
\node[hdr] at (61,-6.2) {coefficient\\[-1pt]$1-m(p)$};
\node[font=\scriptsize,s1col] at (61,-14.3) {3 paths of $s_1$};
\node[coef,fill=obscol!15] at (61,-19.1) {$1-1=0$};
\node[font=\scriptsize,s2col] at (61,-27.3) {1 path of $s_2$};
\node[coef,fill=obscol!15] at (61,-32.1) {$1-1=0$};
\node[font=\scriptsize,gcol!75!black] at (61,-45.6) {6 new paths};
\node[coef,fill=gcol!22,text=gcol!60!black] at (61,-50.4) {$1-0=+1$};
\node[font=\scriptsize,gcol!75!black,align=center,anchor=north] at (61,-53.6) {each crosses\\[-1pt]the interface};
\node[font=\small,anchor=north] at (31,-72.5)
  {$\Delta_\Phi \;=\; 10-3-1 \;=\; 6 \;=\;$ number of new paths};
\begin{scope}[shift={(100,0)}]
\node[panel,anchor=west] at (-20,13) {(b)};
\node[font=\small,anchor=west] at (-13.6,13) {any two parts: the tally};
\def\cA{0}\def\cB{13}\def\cC{26}\def\cD{39}\def\hw{6.1}
\colorlet{shcol}{s1col!50!s2col}
\node[hdr,s1col]          at (\cA,1) {in $s_1$\\[-1pt]only};
\node[hdr,shcol]          at (\cB,1) {in both\\[-1pt](shared)};
\node[hdr,s2col]          at (\cC,1) {in $s_2$\\[-1pt]only};
\node[hdr,gcol!80!black]  at (\cD,1) {in neither\\[-1pt](new)};
\node[font=\scriptsize,anchor=east] at (-7.5,-6)  {$+\Phi(s_1\vee s_2)$};
\node[font=\scriptsize,anchor=east] at (-7.5,-13) {$-\Phi(s_1)$};
\node[font=\scriptsize,anchor=east] at (-7.5,-20) {$-\Phi(s_2)$};
\fill[tbar,obscol!22] (\cA-\hw,-8.2) rectangle (\cD+\hw,-3.8);
\fill[tbar,s1col!25]  (\cA-\hw,-15.2) rectangle (\cB+\hw,-10.8);
\fill[tbar,s2col!28]  (\cB-\hw,-22.2) rectangle (\cC+\hw,-17.8);
\foreach \cx in {\cA,\cB,\cC,\cD} \node[font=\scriptsize] at (\cx,-6) {$+1$};
\foreach \cx in {\cA,\cB} \node[font=\scriptsize] at (\cx,-13) {$-1$};
\foreach \cx in {\cB,\cC} \node[font=\scriptsize] at (\cx,-20) {$-1$};
\draw[black!50,line width=.5pt] (-25,-26) -- (\cD+\hw,-26);
\node[font=\scriptsize,anchor=east] at (-7.5,-31) {$1-m(p)$};
\node[coef,fill=obscol!15,minimum width=9mm] at (\cA,-31) {$0$};
\node[coef,fill=shcol!25,text=shcol!45!black,minimum width=9mm] at (\cB,-31) {$-1$};
\node[coef,fill=obscol!15,minimum width=9mm] at (\cC,-31) {$0$};
\node[coef,fill=gcol!22,text=gcol!60!black,minimum width=9mm] at (\cD,-31) {$+1$};
\node[font=\small,anchor=north] at (12,-39)
  {$\displaystyle\Delta_\Phi=\sum_{p\ \text{new}} q(p)\;-\sum_{p\ \text{shared}} q(p)$};
\node[font=\scriptsize,align=center,anchor=north,text=black!75] at (12,-52)
  {shared paths arise only where the parts overlap;\\
   for disjoint parts the new paths are\\ exactly the boundary-crossing paths};
\end{scope}
\end{tikzpicture}
\caption{\textbf{Scalar emergence is carried by new paths
(Theorem~\ref{thm:scalar}).}  A path $p$ enters $\Delta_\Phi$ with its weight
$q(p)$ times $1-m(p)$, where $m(p)$ is the number of parts that contain $p$.
\textbf{(a)}~Path count for $s_1=(a\to b\to c)$, $s_2=(d\to e)$, and
$\mathcal I=\{(c,d)\}$; each row draws one path under the vertices it visits.
The paths of the parts cancel, and each of the six boundary-crossing paths
contributes $+1$.
\textbf{(b)}~The same tally for arbitrary parts.}
\alttext{Left: a five-vertex chain split into a blue part a to b to c and an
orange part d to e, joined by a dashed red interaction edge from c to d.
Below it are ten rows in three shaded bands, each row drawing one path under
the vertices it visits: the three paths of the blue part and the one path of
the orange part, each band labeled with coefficient zero, and six paths in a
green band, labeled with coefficient plus one, that all cross a vertical
dotted interface line. A line below reads: the discrepancy is ten minus three
minus one, which equals six, the number of new paths. Right: a tally with
four columns, for paths in the first part only, in both parts (shared), in
the second part only, and in neither part (new). A plus-one bar for the whole
spans all four columns, and a minus-one bar for each part spans the two
columns of the paths it contains, giving net coefficients zero, minus one,
zero, and plus one. Below the tally, the discrepancy is written as the total
weight of the new paths minus that of the shared paths.}
\label{fig:scalar}
\end{figure}

We consider observables that sum a weight $q$ over the paths of length at
most $L$; here $\mathcal P_L(G)$, and with it $\mathcal P_\partial$, includes
cycles only when $q$ is defined on them.  Examples are the path count ($q\equiv1$), the
length-discounted count $q(p)=\alpha^{k}$ for a path of length $k$ (a
simple-path analog of a truncated Katz sum), the weighted count
$q(p)=\prod_{e\in p}w(e)$ for fixed edge weights $w$, and the number of paths
that end in a fixed target set.  The identity below also allows parts that
share vertices, with the join given by the formula of
Definition~\ref{def:interaction} and $\mathcal I$ any set of edges in
$(V_1\times V_2)\cup(V_2\times V_1)$ that are not loops.

\begin{theorem}[Scalar emergence is carried by new paths]\label{thm:scalar}
Let $\Phi(G)=\sum_{p\in\mathcal P_L(G)}q(p)$ for every graph $G$, for a
weight $q$ on paths.  Let
$\mathcal P_{\mathrm{new}}$ be the paths of $s_1\vee s_2$ of length at most
$L$ that are paths of neither part, and $\mathcal P_{\mathrm{shared}}$ the
paths of length at most $L$ that belong to both parts.  Then
\[
\Delta_\Phi\;=\;\sum_{p\in\mathcal P_{\mathrm{new}}}q(p)\;-\sum_{p\in\mathcal P_{\mathrm{shared}}}q(p).
\]
If $V_1\cap V_2=\varnothing$, then $\Delta_\Phi=\sum_{p\in\mathcal P_\partial}q(p)$:
scalar emergence is the total weight of the boundary-crossing paths.
\end{theorem}

\begin{proof}
Each part is a subgraph of the join, so every path of a part is a path of
the join, and
$\Delta_\Phi=\sum_{p\in\mathcal P_L(s_1\vee s_2)}\bigl(1-m(p)\bigr)\,q(p)$,
where $m(p)\in\{0,1,2\}$ is the number of parts that contain $p$.  The
coefficient $1-m(p)$ is $+1$ on new paths, $0$ on paths of exactly one
part, and $-1$ on shared paths (Figure~\ref{fig:scalar}).  If
$V_1\cap V_2=\varnothing$, no path lies in both parts, and the new paths
are the boundary-crossing paths by Lemma~\ref{lem:crossing}.
\end{proof}

\begin{corollary}[Path count]\label{cor:pathcount}
For the path count $\Phi(G)=|\mathcal P_L(G)|$ and disjoint parts,
$\Delta_\Phi=|\mathcal P_\partial|$ (take $q\equiv1$ in
Theorem~\ref{thm:scalar}).  Under the graph power $G^{\le L}$ every path
produces an edge, since its end vertices differ, so
$\mathcal S_\partial=\mathcal P_\partial$ and the path-count emergence
equals the bound $N_\partial$ for $G^{\le L}$.
\end{corollary}

For the weighted count with positive edge weights between $w_{\min}$ and
$w_{\max}$ and disjoint parts, every boundary-crossing path has weight between
$\min(w_{\min},w_{\min}^{L})$ and $\max(w_{\max},w_{\max}^{L})$, so its scalar
discrepancy $\Delta_\Phi$ lies between $|\mathcal P_\partial|$ times these two
numbers.

The same boundary-crossing paths thus account exactly for the emergence of
observables that sum a weight over paths and bound the emergence of every
observation acting on paths, one route per emergent edge.

\section{What produces emergence}\label{sec:examples}

Theorem~\ref{thm:main} places every emergent edge on a route across the
interface.  We now see which observations read such routes
(Figure~\ref{fig:examples}) and how cutting or adding interaction edges
suppresses or promotes emergence.

\begin{figure}[htbp]
\centering
\begin{adjustbox}{max width=\linewidth}\begin{tikzpicture}[x=1cm,y=1cm,
  vtx/.style={circle,draw,thick,minimum size=5.6mm,inner sep=0pt,font=\footnotesize},
  hd/.style={font=\scriptsize,align=center,anchor=south},
  rt/.style={font=\scriptsize,anchor=west,inner sep=0pt},
  wt/.style={font=\scriptsize,anchor=west,align=left,inner sep=0pt},
  sep/.style={obscol!25},
  blk/.style={draw=obscol!70,dashed,rounded corners=4pt,inner sep=1.6pt},
  blklbl/.style={font=\scriptsize,text=obscol,inner sep=1pt}]

\def\cA{1.5}\def\cB{4.65}\def\cC{7.8}\def\cD{10.95}\def\cW{12.7}
\def\sH{1.2}\def\sT{0.6}

\node[hd] at (\cA,0.02)  {\textbf{parts}\\[1pt]$s_1,\ s_2$};
\node[hd] at (\cB,0.02)  {\textbf{join}\\[1pt]$s_1\vee s_2$};
\node[hd] at (\cC,0.02)  {\textbf{observed union}\\[1pt]$\Phi(s_1)\vee\Phi(s_2)$};
\node[hd] at (\cD,0.02) {\textbf{observed whole}\\[1pt]$\Phi(s_1\vee s_2)$};
\node[hd] at (14.15,0.02)  {\textbf{witness paths}\\[1pt]\textbf{and the bound}};
\draw[obscol!45] (0,-0.05) -- (15.7,-0.05);

\def\yt{-0.36}\def\yc{-1.52}
\node[rt] at (0,\yt) {\textbf{(a) Drop sinks.}\ An edge is kept when its head can go on, that is, when the head has an out-edge.};
\foreach \xx in {3.075,6.225,9.375,12.525}{\draw[sep] (\xx,\yc+0.9) -- (\xx,\yc-0.9);}
\begin{scope}[shift={(\cA-\sH/2,\yc)}]
  \node[v1](u) at (0,\sT){$u$}; \node[v1,vb](a) at (\sH,\sT){$a$}; \node[v2,vb](v) at (\sH,-\sT){$v$};
  \draw[eobs](u)--(a);
\end{scope}
\begin{scope}[shift={(\cB-\sH/2,\yc)}]
  \node[v1](u) at (0,\sT){$u$}; \node[v1,vb](a) at (\sH,\sT){$a$}; \node[v2,vb](v) at (\sH,-\sT){$v$};
  \draw[eobs](u)--(a); \draw[ei](a)--(v);
\end{scope}
\begin{scope}[shift={(\cC-\sH/2,\yc)}]
  \node[v1](u) at (0,\sT){$u$}; \node[v1,vb](a) at (\sH,\sT){$a$}; \node[v2,vb](v) at (\sH,-\sT){$v$};
  \draw[ex](u)--(a); \draw[ei](a)--(v);
\end{scope}
\begin{scope}[shift={(\cD-\sH/2,\yc)}]
  \node[v1](u) at (0,\sT){$u$}; \node[v1,vb](a) at (\sH,\sT){$a$}; \node[dead](v) at (\sH,-\sT){$v$};
  \draw[eg](u)--(a); \draw[el](a)--(v);
\end{scope}
\node[wt] at (\cW,\yc) {$(u,a,v)\mapsto\textcolor{gcol}{(u,a)}$\\[3pt]
  $|E^+_\Phi|=1=N_\partial$};

\def\yt{-2.80}\def\yc{-3.96}
\node[rt] at (0,\yt) {\textbf{(b) Edges on a cycle ($L=4$).}\ An edge is kept when a route of length at most $3$ leads from its head back to its tail.};
\foreach \xx in {3.075,6.225,9.375,12.525}{\draw[sep] (\xx,\yc+0.9) -- (\xx,\yc-0.9);}
\begin{scope}[shift={(\cA-\sH/2,\yc)}]
  \node[v1,vb](a) at (0,\sT){$a$}; \node[v1,vb](b) at (\sH,\sT){$b$};
  \node[v2,vb](c) at (\sH,-\sT){$c$}; \node[v2,vb](d) at (0,-\sT){$d$};
  \draw[eobs](a)--(b); \draw[eobs](c)--(d);
\end{scope}
\begin{scope}[shift={(\cB-\sH/2,\yc)}]
  \node[v1,vb](a) at (0,\sT){$a$}; \node[v1,vb](b) at (\sH,\sT){$b$};
  \node[v2,vb](c) at (\sH,-\sT){$c$}; \node[v2,vb](d) at (0,-\sT){$d$};
  \draw[eobs](a)--(b); \draw[eobs](c)--(d); \draw[ei](b)--(c); \draw[ei](d)--(a);
\end{scope}
\begin{scope}[shift={(\cC-\sH/2,\yc)}]
  \node[v1,vb](a) at (0,\sT){$a$}; \node[v1,vb](b) at (\sH,\sT){$b$};
  \node[v2,vb](c) at (\sH,-\sT){$c$}; \node[v2,vb](d) at (0,-\sT){$d$};
  \draw[ex](a)--(b); \draw[ex](c)--(d); \draw[ei](b)--(c); \draw[ei](d)--(a);
\end{scope}
\begin{scope}[shift={(\cD-\sH/2,\yc)}]
  \node[v1,vb](a) at (0,\sT){$a$}; \node[v1,vb](b) at (\sH,\sT){$b$};
  \node[v2,vb](c) at (\sH,-\sT){$c$}; \node[v2,vb](d) at (0,-\sT){$d$};
  \draw[eg](a)--(b); \draw[eg](c)--(d); \draw[ei](b)--(c); \draw[ei](d)--(a);
\end{scope}
\node[wt] at (\cW,\yc) {$(a,b,c,d,a)\mapsto\textcolor{gcol}{(a,b)}$\\[1pt]
  $(c,d,a,b,c)\mapsto\textcolor{gcol}{(c,d)}$\\[3pt]
  $|E^+_\Phi|=2\le N_\partial=4$};

\def\yt{-5.24}\def\yc{-6.40}
\node[rt] at (0,\yt) {\textbf{(c) Coarse-grained power.}\ Paths of length $\le 2$ link the fixed blocks $X=\{x,y\}$, $A=\{a\}$, and $B=\{b\}$.};
\foreach \xx in {3.075,6.225,9.375,12.525}{\draw[sep] (\xx,\yc+0.9) -- (\xx,\yc-0.9);}
\begin{scope}[shift={(\cA-\sH/2+0.1,\yc)}]
  \node[v1](x) at (0,\sT){$x$}; \node[v1,vb](a) at (\sH,\sT){$a$};
  \node[v2,vb](b) at (\sH,-\sT){$b$}; \node[v2](y) at (0,-\sT){$y$};
  \draw[eobs](x)--(a); \draw[eobs](b)--(y);
  \node[blk,fit=(x)(y)](X){}; \node[blklbl,anchor=east] at (X.west){$X$};
\end{scope}
\begin{scope}[shift={(\cB-\sH/2+0.1,\yc)}]
  \node[v1](x) at (0,\sT){$x$}; \node[v1,vb](a) at (\sH,\sT){$a$};
  \node[v2,vb](b) at (\sH,-\sT){$b$}; \node[v2](y) at (0,-\sT){$y$};
  \draw[eobs](x)--(a); \draw[eobs](b)--(y); \draw[ei](a)--(b);
  \node[blk,fit=(x)(y)](X){}; \node[blklbl,anchor=east] at (X.west){$X$};
\end{scope}
\begin{scope}[shift={(\cC-\sH/2,\yc)}]
  \node[vn](X) at (0,0){$X$}; \node[v1,vb](A) at (\sH,\sT){$A$}; \node[v2,vb](B) at (\sH,-\sT){$B$};
  \draw[eobs](X) to[bend left=18] (A); \draw[eobs](B) to[bend left=18] (X);
  \draw[ei](A)--(B);
\end{scope}
\begin{scope}[shift={(\cD-\sH/2,\yc)}]
  \node[vn](X) at (0,0){$X$}; \node[v1,vb](A) at (\sH,\sT){$A$}; \node[v2,vb](B) at (\sH,-\sT){$B$};
  \draw[eobs](X) to[bend left=18] (A); \draw[eobs](B) to[bend left=18] (X);
  \draw[ei](A)--(B);
  \draw[eg](A) to[bend left=18] (X); \draw[eg](X) to[bend left=18] (B);
\end{scope}
\node[wt] at (\cW,\yc) {$(x,a,b)\mapsto\textcolor{gcol}{(X,B)}$\\[1pt]
  $(a,b,y)\mapsto\textcolor{gcol}{(A,X)}$\\[3pt]
  $|E^+_\Phi|=2\le N_\partial=3$\\[1pt]
  sharpened bound $2$, attained};

\end{tikzpicture}\end{adjustbox}
\caption{\textbf{Observations that read routes produce emergence, and
boundary-crossing paths account for it
(Examples~\ref{ex:dropsinks}--\ref{ex:coarse}).}
Each row shows the parts, their join, the observed union, and the observed
whole; the last column pairs each emergent edge (green) with its witness path
and gives the bound (in (c), also the sharpened bound of
Corollary~\ref{cor:sharp}).  The interaction edges already lie in the observed
union; in (a) the observed whole lacks the interaction edge, which is lost
(dotted purple).  Blue, orange: $s_1$, $s_2$; dashed red: $\mathcal{I}$; double rings:
$\partial$; faded: deleted by $\Phi$; dashed box and gray vertex: block $X$.}
\alttext{Three rows, one per observation: drop sinks, edges on a cycle, and a coarse-grained power. Each row shows two small parts, their join through dashed red interaction edges, the observed union, and the observed whole, with emergent edges drawn in green. A final column lists, for each emergent edge, the boundary-crossing path that produces it, and the bound, which compares the number of emergent edges with the number of boundary-crossing paths that produce an edge: one emergent edge and one path in the first row, two emergent edges and four paths in the second, and two emergent edges and three paths, two of them of length two, in the third.}
\label{fig:examples}
\end{figure}

\subsection{Emergence requires routes}

\begin{proposition}[Emergence requires routes]\label{prop:routes}
If $\Phi$ decides each observed edge from a single edge of the graph, that is,
if $\Phi$ acts on paths of length at most one, then $E^+_\Phi=\varnothing$ for
all parts $s_1,s_2$ and every interaction~$\mathcal I$.
\end{proposition}

\begin{proof}
With $L=1$, every boundary-crossing path has length one, so
Corollary~\ref{cor:sharp} gives $E^+_\Phi=\varnothing$.
\end{proof}

Thresholding a weighted network, such as a brain connectivity
matrix~\cite{bullmore2009complex}, and coarse-graining, a standard way to view
a network at several
scales~\cite{song2005selfsimilarity,itzkovitz2005coarse,gfeller2007spectral,garciaperez2018multiscale},
are of this kind when the threshold or the partition is fixed in advance: every
edge they show in the whole already lies in the observed parts joined by their
interaction, even when a block straddles the interface.  For observations
acting on paths, emergence therefore comes from routes of length at least two
(Corollary~\ref{cor:sharp}).  Graph powers and their compositions, and the
deletion rules of Section~\ref{sec:framework} that read context, use such
routes: their verdict on an edge depends on the
routes around it, and the join supplies routes that the parts lack.

\subsection{Observations that read routes}

\begin{example}[Drop sinks]\label{ex:dropsinks}
Let $\Phi$ drop sinks, and take $s_1=(u\to a)$, $s_2=(\{v\},\varnothing)$, and
$\mathcal I=\{(a,v)\}$ (Figure~\ref{fig:examples}(a)).  In $s_1$ the vertex
$a$ is a sink, so $\Phi(s_1)$ has no edge; in the join, $a$ gains the out-edge
$(a,v)$, and $(u,a)$ is kept.  The vertex $v$ is a sink in the join as well,
so the interaction edge is deleted.  Hence $E^+_\Phi=\{(u,a)\}$ and
$E^-_\Phi=\{(a,v)\}$, as Proposition~\ref{prop:lost} predicts.  The only
boundary-crossing path that produces an edge is $(u,a,v)$, so
$|E^+_\Phi|=1=N_\partial$: the bound is attained.
\end{example}

\begin{example}[Edges on a cycle]\label{ex:cycle}
Let $\Phi$ keep the edges that lie on a directed cycle of length at most
$L=4$, and take $s_1=(a\to b)$, $s_2=(c\to d)$, and
$\mathcal I=\{(b,c),(d,a)\}$ (Figure~\ref{fig:examples}(b)).  Neither part
has a cycle, so the observed union has only the two interaction edges.
The join closes the four-cycle $(a,b,c,d,a)$, so $\Phi$ keeps every edge:
$E^+_\Phi=\{(a,b),(c,d)\}$ and $E^-_\Phi=\varnothing$.  The boundary-crossing
cycles are the four rotations of this cycle, so $N_\partial=4$.  The rotations
starting at $a$ and at $c$ witness the emergent edges, and the other two
produce the interaction edges, which the observed union already contains:
$|E^+_\Phi|=2\le N_\partial=4$.  One feedback loop closed through the
interface makes the edges of both parts visible at once; its witnesses are
found by enumerating simple cycles~\cite{johnson1975finding}.
\end{example}

\begin{example}[Coarse-grained power]\label{ex:coarse}
Let $\Phi$ be the graph power $G^{\le2}$ followed by coarse-graining by the
fixed blocks $X=\{x,y\}$, $A=\{a\}$, and $B=\{b\}$.  Take $s_1=(x\to a)$,
$s_2=(b\to y)$, and $\mathcal I=\{(a,b)\}$ (Figure~\ref{fig:examples}(c)); the
block $X$ straddles the interface.  The observed union has the edges $(X,A)$,
$(B,X)$, and the image $(A,B)$ of the interaction edge.  In the join, the paths
$(x,a,b)$ and $(a,b,y)$ add $E^+_\Phi=\{(X,B),(A,X)\}$, and
$E^-_\Phi=\varnothing$.  The boundary-crossing paths are $(a,b)$, $(x,a,b)$,
and $(a,b,y)$, so $N_\partial=3$; the interaction edge produces its own image,
and the sharpened bound of Corollary~\ref{cor:sharp}, which counts the two
paths of length two, is attained.  Applied edge by edge, the same partition shows
only edges of the observed union (Proposition~\ref{prop:routes}); reading
routes first reveals new links between blocks, which
Section~\ref{sec:granularity} follows as the partition is refined.  Causal
emergence is likewise measured on coarse-grained
descriptions~\cite{hoel2017map}, including macro-nodes of
networks~\cite{klein2020emergence}.
\end{example}

\subsection{Suppressing and promoting emergence}

Read in reverse, Theorem~\ref{thm:main} says that if no route across the
interface produces an edge, the observation sees the parts side by side.

\begin{proposition}[When no boundary-crossing path produces an edge]\label{prop:cut}
Let $\Phi$ act on paths of length at most $L$.  If no boundary-crossing path
produces an edge, that is, $N_\partial=0$, then
$E(\Phi(s_1\vee s_2))=E(\Phi(s_1))\cup E(\Phi(s_2))$.  Hence
$E^+_\Phi=\varnothing$, and the lost edges are the images of the interaction
edges that the observed parts do not show,
$E^-_\Phi=\varphi(\mathcal I)\setminus\bigl(E(\Phi(s_1))\cup E(\Phi(s_2))\bigr)$.
\end{proposition}

\begin{proof}
With $\mathcal S_\partial=\varnothing$, Theorem~\ref{thm:main}(i) gives
$E(\Phi(s_1\vee s_2))=E(\Phi(s_1))\cup E(\Phi(s_2))$, and the formulas for
$E^\pm_\Phi$ follow from Definition~\ref{def:discrepancy}.
\end{proof}

\begin{example}[Thresholding away the coupling]\label{ex:cut}
Take $s_1=(u\to a)$, $s_2=(v\to w)$, and $\mathcal I=\{(a,v)\}$, with weights
$0.9$, $0.9$, and $0.1$, and threshold at $\tau=0.5$.  The light interaction
edge produces nothing, and it is the only boundary-crossing path, so
$N_\partial=0$.  The observed whole shows the parts side by side:
$E^+_\Phi=\varnothing$, and the interaction edge is lost,
$E^-_\Phi=\{(a,v)\}$.  The thresholded power gives the same result, because
every route across the interface uses the light edge.
\end{example}

Proposition~\ref{prop:cut} treats the case in which no boundary-crossing path
produces an edge.
When only some interaction edges are removed, the effect can be followed
emergent edge by emergent edge.  For $\mathcal J\subseteq\mathcal I$, write
$E^+_\Phi(\mathcal J)$ and
$U(\mathcal J):=E(\Phi(s_1))\cup E(\Phi(s_2))\cup\varphi(\mathcal J)$ for the
emergent edges and the observed union of the parts joined along $\mathcal J$,
so that $E^+_\Phi=E^+_\Phi(\mathcal I)$ and $U=U(\mathcal I)$.

\begin{proposition}[Removing interaction edges]\label{prop:intervene}
In the setting of Theorem~\ref{thm:main}, let $C\subseteq\mathcal I$ and
$\mathcal J:=\mathcal I\setminus C$.
\begin{enumerate}[label=(\roman*)]
\item The boundary-crossing paths of $s_1\vee_{\mathcal J}s_2$ that produce an
edge are the paths of $\mathcal S_\partial$ that traverse no edge of $C$.
\item An emergent edge $e\in E^+_\Phi$ remains emergent after the removal of
$C$ if and only if some path of $\mathcal S_\partial(e)$ traverses no edge of
$C$.  Equivalently, removing $C$ removes exactly the emergent edges $e$ for
which every path in $\mathcal S_\partial(e)$ traverses an edge of $C$.
\item The new emergent edges, those of $E^+_\Phi(\mathcal J)\setminus
E^+_\Phi$, are exactly the edges of $\varphi(C)\setminus U(\mathcal J)$
produced by a path of~(i): images of removed interaction edges that the
reduced observed union no longer contains and that a boundary-crossing path
avoiding $C$ still produces.
\end{enumerate}
In particular, the number of emergent edges removed, and hence
$|E^+_\Phi|-|E^+_\Phi(\mathcal J)|$, is at most the credit of~$C$.
\end{proposition}

\begin{proof}
Removing $C$ keeps every vertex and removes only the edges of $C$, so the paths
(and cycles) of $s_1\vee_{\mathcal J}s_2$ are those of $s_1\vee s_2$ that
traverse no edge of $C$; this gives~(i).  Let $\mathcal S'$ be the set of paths
in~(i).  By Theorem~\ref{thm:main}(i) for both interactions,
$E^+_\Phi=\prodedge{\mathcal S_\partial}\setminus U$ and
$E^+_\Phi(\mathcal J)=\prodedge{\mathcal S'}\setminus U(\mathcal J)$, where
$\prodedge{\mathcal S'}\subseteq\prodedge{\mathcal S_\partial}$ and
$U(\mathcal J)\subseteq U$.  Hence
\[
E^+_\Phi\cap E^+_\Phi(\mathcal J)=\prodedge{\mathcal S'}\setminus U
\qquad\text{and}\qquad
E^+_\Phi(\mathcal J)\setminus E^+_\Phi=\prodedge{\mathcal S'}\cap\bigl(U\setminus U(\mathcal J)\bigr).
\]
An emergent edge $e$ lies in $\prodedge{\mathcal S'}$ exactly when some path of
$\mathcal S_\partial(e)$ traverses no edge of $C$, which gives~(ii); and
$U\setminus U(\mathcal J)=\varphi(C)\setminus U(\mathcal J)$, which
gives~(iii).  Finally, for each removed emergent edge $e$, every path of
$\mathcal S_\partial(e)$ traverses an edge of $C$ and their shares sum to $1$,
so the removed edges number at most the credit of~$C$.
\end{proof}

For $C=\mathcal I$, every emergent edge is removed, since every
boundary-crossing path traverses an edge of $\mathcal I$.  The edges in (iii) are images of removed
interaction edges that the reduced join still produces along
boundary-crossing paths of length at least two: without the direct coupling, the observed whole reveals
them as emergent.

Together with Theorem~\ref{thm:main}, Propositions~\ref{prop:cut}
and~\ref{prop:intervene} show how to suppress emergence.  To suppress an
emergent edge, cut every boundary-crossing path that produces it; to suppress all emergence under a given
observation, remove every boundary-crossing path that produces an edge: delete
interaction edges or weaken them below the observation's threshold, or break
the internal routes that feed them, all of which lie within $L-1$ steps of
$\partial$ (Corollary~\ref{cor:local}).  Such cuts only remove paths from
$\mathcal S_\partial$, so each lowers the bound $N_\partial$ or leaves it
unchanged, and once $N_\partial$ reaches zero, Proposition~\ref{prop:cut}
applies.  The credit of a set of interaction edges bounds how many emergent
edges its removal destroys (Proposition~\ref{prop:intervene}), which indicates
where to cut.

To promote emergence, add interaction edges between internally structured
regions.  An interaction edge $(a,v)$ opens a boundary-crossing path for each
route into $a$ combined with each route out of $v$, of total length at most
$L$, so one bridge raises the bound $N_\partial$ by many paths at once
(Proposition~\ref{prop:bridge} and Example~\ref{ex:bowtie}), much as
cross-module connections integrate specialized modules in brain
networks~\cite{sporns2010networks};
Section~\ref{sec:bridging} shows emergence rising with the bound.

\section{Numerical experiments}\label{sec:numerics}

The experiments compute both sides of Theorem~\ref{thm:main} exactly. The bound
holds in every instance and, beyond the interaction edges, is attained where the
routes through the interface are distinct
(Sections~\ref{sec:validity}--\ref{sec:tightness}), as they become in large
sparse networks (Section~\ref{sec:scale}). Measured emergence matches the
exact expectation given by the bridge law (Section~\ref{sec:bridging}); for
observations that read three or more steps, $N_\partial$ tracks emergence more
closely than the cut size, the degrees at the interface, and edge betweenness
(Section~\ref{sec:baselines}); and on real directed networks the shares locate
the components that carry emergence, while Proposition~\ref{prop:intervene}
gives the exact effect of removing interaction edges (Section~\ref{sec:real}).

\subsection{Method}\label{sec:method}

In Experiments~1--6 and~8, whose parts have at most a few tens of vertices, we
enumerate the simple directed paths (and, for drop sinks and edges on a cycle,
the simple cycles) of length at most $L$ in the parts and in their join. Each
observation is
implemented by its path rule (Definition~\ref{def:path_obs}), checked against
its direct definition (Appendix~\ref{app:experiments}); $E^+_\Phi$ comes from
comparing $\Phi(s_1\vee s_2)$ with the observed union, and $N_\partial$ from a
separate enumeration of the producing boundary-crossing paths. Experiments~7
and~9 use the local search of Corollary~\ref{cor:local}, checked against a
direct computation. Random networks come from
NetworkX~\cite{hagberg2008exploring} or from their defining rules. We report means
with normal-approximation $95\%$ confidence intervals (bootstrap intervals for
Experiment~8); all randomness is seeded, and code that reproduces every
experiment and figure is provided \codestatement.

\subsection{Validity across observations and network families}\label{sec:validity}

\paragraph{Experiment~1: thirteen observations.} Two random digraphs on five
vertices, in which each ordered pair is an edge independently with probability
$p\in\{0.15,0.30,0.45\}$ (the directed $G(n,p)$
variant~\cite{newman2018networks} of the Erd\H{o}s--R\'enyi
model~\cite{erdos1959random}), are joined by $|\mathcal I|\in\{1,2,3\}$ random
interaction edges, $30$ times per combination, under thirteen observations.
These form three groups: the edgewise observations (identity, thresholding, and
coarse-graining by a fixed partition), \emph{powers} that read routes
($G^{\le L}$ for $L=2,3,4$, and coarse-grained and thresholded powers), and
\emph{deletion rules that read context} (drop sinks, edges on a cycle of length
at most $4$, and the continuation filter with $L=2,3$). This gives $3510$
instances (Table~\ref{tab:exp1} in Appendix~\ref{app:experiments}).

\begin{figure}[htbp]
\centering
\includegraphics[width=\textwidth]{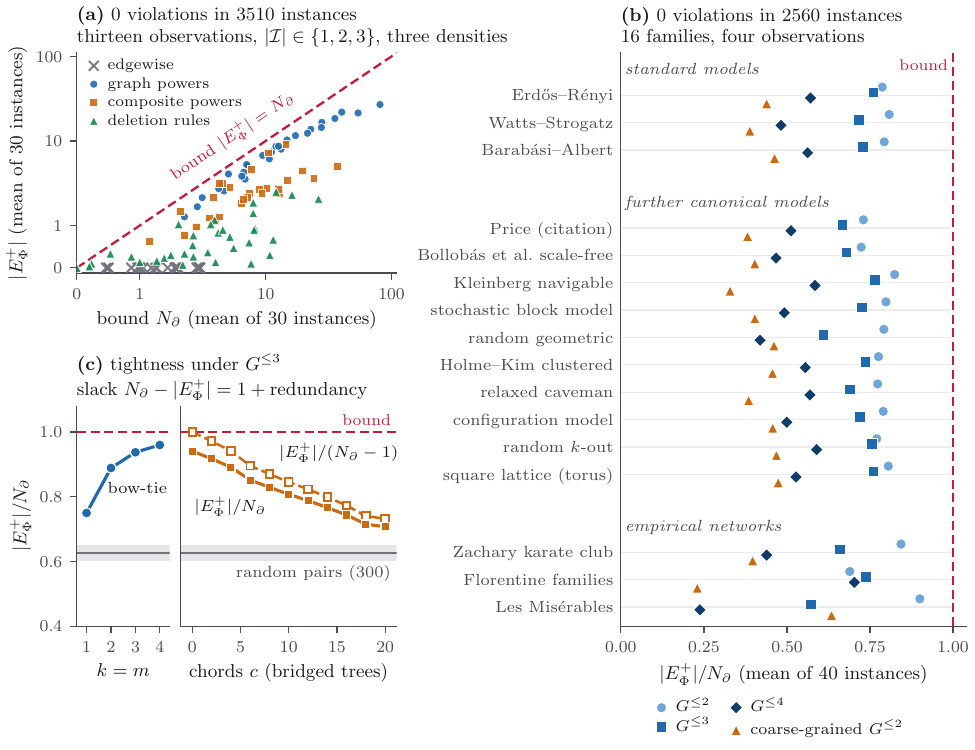}
\caption{\textbf{The bound holds across observations and network families, and is approached where routes are distinct.}
\textbf{(a)}~Experiment~1: mean $|E^+_\Phi|$ against mean $N_\partial$ over the $30$ instances of each combination of observation, edge probability, and $|\mathcal I|$ ($117$ points, $3510$ instances); axes linear on $[0,1]$ and logarithmic above; composite powers are the coarse-grained and thresholded powers.
\textbf{(b)}~Experiment~2: mean ratio $|E^+_\Phi|/N_\partial$ over $40$ instances for sixteen families and four observations ($2560$ instances).
\textbf{(c)}~Experiment~3 under $G^{\le3}$ with one interaction edge: the ratio on bow-ties (left) and on two bridged binary trees with $c$ random chords added inside the parts (right; filled, $|E^+_\Phi|/N_\partial$; open, the ratio to the sharpened bound, $|E^+_\Phi|/(N_\partial-1)$; $100$ realizations per $c$), against $300$ random pairs (gray band). Bars and bands are $95\%$ confidence intervals.}
\alttext{Three panels. Top left: a scatter of emergent edges against the bound, on axes that are linear from zero to one and logarithmic above, with every point on or below a dashed red diagonal and the edgewise observations on the zero line. Right: a dot plot listing sixteen network families in three groups, each with four markers for the four observations, all to the left of a dashed red line at ratio one. Bottom left, two plots sharing one vertical axis: the ratio of emergent edges to the bound rising toward one for bow-ties of growing size, and, for trees with a growing number of chords, a curve of filled squares starting at 0.94 and a curve of open squares starting at one, both decreasing to about 0.71 and 0.73 as chords are added; a gray band for random pairs near 0.63 runs across both.}
\label{fig:exp_valid}
\end{figure}

The bound holds in all $3510$ instances (Figure~\ref{fig:exp_valid}(a)).
Emergence requires routes: under the edgewise observations, $|E^+_\Phi|=0$ in
all $810$ of their instances, as Proposition~\ref{prop:routes} predicts, and the
only lost edges are the light interaction edges that thresholding removes
(Proposition~\ref{prop:lost}). Observations that read routes produce emergence
in abundance, growing with density and reach for the graph powers: $1804$
instances show emergent edges, among them $19$ in which a feedback loop closed
through the interface reveals edges of the parts (Example~\ref{ex:cycle}). The
bound is attained in $37$ of the $3107$ instances with $N_\partial>0$, all under
deletion rules ($31$ under the continuation filter and $6$ under drop sinks, as
in Example~\ref{ex:dropsinks}). The sharpened bound of
Corollary~\ref{cor:sharp} is attained in $533$ of the $2290$ instances with a
boundary-crossing path of length at least two.

\paragraph{Experiment~2: sixteen network families.} Synthetic parts of $12$
vertices come from the Erd\H{o}s--R\'enyi, Watts--Strogatz, and
Barab\'asi--Albert
models~\cite{erdos1959random,watts1998collective,barabasi1999emergence} and
from ten further canonical models that add structural features these three
lack; the other parts are random halves of three empirical networks (karate
club~\cite{zachary1977information}, Florentine
families~\cite{padgett1993robust,breiger1986cumulated}, and \emph{Les
Mis\'erables}~\cite{knuth1993stanford}). The joins are observed with
$G^{\le2}$, $G^{\le3}$, $G^{\le4}$, and coarse-grained $G^{\le2}$, with
$|\mathcal I|=3$ and $40$ realizations per cell ($2560$ instances;
Appendix~\ref{app:experiments}). The bound and the sharpened bound
hold in all $2560$ (Figure~\ref{fig:exp_valid}(b)). The mean of
$|E^+_\Phi|/N_\partial$ is set mainly by the observation: across the $13$
synthetic families it falls with the reach, as more routes join the same
endpoints, and lies in $[0.72,0.82]$ for $G^{\le2}$, $[0.61,0.77]$ for
$G^{\le3}$, $[0.42,0.59]$ for $G^{\le4}$, and $[0.33,0.47]$ for coarse-grained
$G^{\le2}$; over all $16$ families it ranges from $0.23$ (Florentine families,
coarse-grained) to $0.90$ (\emph{Les Mis\'erables}, $G^{\le2}$).

\paragraph{The identity behind the bound.} On every instance of Experiments~1
and~2, on the $1305$ bow-ties, trees, and random pairs of Experiment~3 below,
and on $1560$ random joins of three or four parts, each of four or five
vertices (Remark~\ref{rem:multipart}), we verified the identity and the
equality criterion of Theorem~\ref{thm:main}, the sharpened bound,
Lemma~\ref{lem:crossing}, Proposition~\ref{prop:lost}, the total-discrepancy
bound of Corollary~\ref{cor:discrepancy}, and that the local search of
Corollary~\ref{cor:local} returns exactly the boundary-crossing paths and
cycles found by enumerating the join. All of them hold on all $8935$ instances.

\subsection{Tightness}\label{sec:tightness}

\paragraph{Experiment~3: tightness under $G^{\le3}$.} Three designs share one
interaction edge $(a,v)$ from $s_1$ to $s_2$ (Figure~\ref{fig:exp_valid}(c)).
The path $(a,v)$ produces an edge that the observed union shows, and every
longer boundary-crossing path produces an edge from $V_1$ to $V_2$ that the
observed union lacks. Hence, by Theorem~\ref{thm:main}(i), the slack is
$N_\partial-|E^+_\Phi|=1+\bigl(|\mathcal S_\partial|-|\prodedge{\mathcal S_\partial}|\bigr)$,
one plus the \emph{redundancy}, the number of boundary-crossing paths in excess
of one per produced edge. On the
bow-tie of Example~\ref{ex:bowtie} with
$k=m=1,\dots,4$, all routes are distinct: $|E^+_\Phi|=3,8,15,24$ and
$N_\partial=4,9,16,25$, a ratio rising from $0.75$ to $0.96$. The same holds for a
binary in-tree of depth two into $a$ bridged to a binary out-tree of depth two
out of $v$: $|E^+_\Phi|=16=N_\partial-1$, the equality case of
Proposition~\ref{prop:bridge}(iii). Adding $c=2,4,\dots,20$ random chords
inside these parts ($100$ realizations each) opens alternative routes, and the
slack tracks them exactly. At $c=20$, the bound, the emergent edges, and the
redundancy average $31.4$, $22.0$, and $8.4$, with a mean ratio of
$0.71\pm0.01$ and a mean ratio to the sharpened bound of $0.73\pm0.01$. In all $1001$
realizations, the slack is one plus the redundancy. On $300$ random pairs drawn
as in Experiment~1 (four vertices per part, $p=0.3$), the mean ratio is
$0.63\pm0.03$; over the $275$ with a boundary-crossing path of length at
least two, the ratio to the sharpened bound averages $0.93\pm0.01$, and the
sharpened bound
is attained in $192$. On $2400$ further joins of two four-vertex parts (edge
probability $0.35$, one to three interaction edges in either direction) under
$G^{\le2}$ and $G^{\le3}$, the sharpened bound holds in all $4800$ instances.
Over the $4656$ with a boundary-crossing path of length at least two, sharpening
raises the mean ratio from $0.62$ to $0.86$, and the sharpened bound is
attained in $2325$. Beyond
the interaction edges, the slack thus measures mainly how often the observation
merges distinct routes into one edge.

\subsection{Topology and reach}\label{sec:topology}

\begin{figure}[htbp]
\centering
\includegraphics[width=\textwidth]{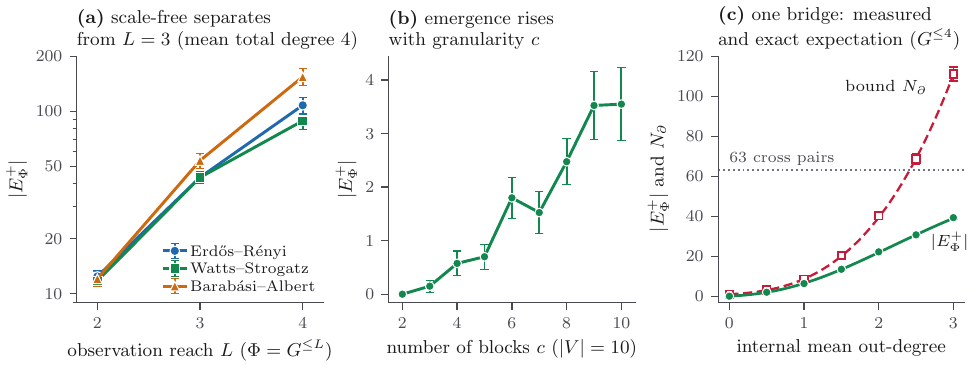}
\caption{\textbf{What emergence depends on.}
\textbf{(a)}~Experiment~4: Erd\H{o}s--R\'enyi, Watts--Strogatz, and Barab\'asi--Albert parts of $40$ vertices (the last two oriented as in Experiment~2) at matched mean total degree $4$, $\Phi=G^{\le L}$, $60$ realizations per point.
\textbf{(b)}~Experiment~5: coarse-grained $G^{\le2}$ on $10$ vertices with a balanced random partition into $c$ blocks, $40$ realizations per point.
\textbf{(c)}~Experiment~6: two modules of eight vertices joined by a single interaction edge, $\Phi=G^{\le4}$, $1000$ realizations per point; emergent edges (circles) and the bound $N_\partial$ (squares) against the internal mean out-degree, with their exact expectations (curves, Corollary~\ref{cor:bridge_random}) and the $63$ cross pairs that can become emergent (dotted line).
Error bars are $95\%$ confidence intervals.}
\alttext{Three panels. Left: emergent edges, on a logarithmic scale, against observation reach two to four for three network families; the three curves start together at reach two and the scale-free curve rises above the other two from reach three. Middle: emergent edges rising with the number of blocks of a coarse-graining, from zero at two blocks to about three and a half at ten. Right: measured emergent edges and bound as markers lying on two smooth curves of exact expectations, both rising with the internal mean out-degree of two modules joined by one edge; the bound curve rises to about 111 and the emergent-edge curve to about 39, below a dotted line at 63.}
\label{fig:exp_depend}
\end{figure}

\paragraph{Experiment~4.} To isolate topology, we draw Erd\H{o}s--R\'enyi,
Watts--Strogatz, and Barab\'asi--Albert parts of $40$
vertices~\cite{erdos1959random,watts1998collective,barabasi1999emergence}, the
last two oriented as in Experiment~2, and subsample their edges so that, within
each density, all three families have exactly the same number of edges, at mean
total degree $2$, $4$, and $6$. The observation is $G^{\le L}$ with reach
$L=2,3,4$; three interaction edges run from $s_1$ to $s_2$, and each cell has
$60$ realizations. At reach $L=2$, the three ensembles coincide within their
$95\%$ intervals at every density; at mean total degree $4$, they give
$12.45\pm0.96$, $11.78\pm0.63$, and $12.10\pm1.12$ emergent edges, as
Theorem~\ref{thm:main} leads one to expect. At reach two, an emergent edge
extends an interaction edge by one internal edge, so emergence is set by the
degrees of the boundary vertices, fixed on average by the matched edge count.
With longer reach, topology can enter: at mean total degree $4$, the
scale-free (Barab\'asi--Albert) ensemble separates upward
from $L=3$ ($53.7\pm5.1$ against $43.4\pm3.3$ and $43.4\pm3.1$) and reaches
$1.43$ times the Erd\H{o}s--R\'enyi value at $L=4$ ($154.2\pm16.5$ against
$107.7\pm11.4$; Figure~\ref{fig:exp_depend}(a)). A vertex reached along an edge
has, on average, more than the average
degree~\cite{newman2001random,newman2003structure} (in the randomly oriented
Barab\'asi--Albert parts, also more than the average out-degree), so a longer
route that enters a hub fans out, which accounts for the separation at this
density.

\subsection{Coarse-graining granularity}\label{sec:granularity}

\paragraph{Experiment~5.} Coarse-graining by a fixed partition alone reproduces
the observed union
(Proposition~\ref{prop:routes}); composed with $G^{\le2}$, it reads routes and
produces emergence (Example~\ref{ex:coarse}). On joins of two five-vertex parts
(edge probability $0.25$, two interaction edges from $s_1$ to $s_2$), we draw a
balanced random partition of the $10$ vertices into $c=2,\dots,10$ blocks, free
to straddle the parts, and observe with coarse-grained $G^{\le2}$, with $40$
realizations per value of $c$. Emergence rises with granularity, from $0$ at
$c=2$ to $3.55$ at $c=10$, and is highest at the two finest partitions
(Figure~\ref{fig:exp_depend}(b)); the one dip, from $1.80$ at $c=6$ to $1.53$
at $c=7$, lies within the confidence intervals. A coarser partition folds more
boundary-crossing paths into a single block or onto block pairs that the
observed union already contains: the scale of observation sets how much of the
emergence carried by the interface becomes visible.

\subsection{Bridging internally structured modules}\label{sec:bridging}

When two modules are joined by a single interaction edge, the bridge law
(Proposition~\ref{prop:bridge}) gives the boundary-crossing paths and the
emergent edges from the reach of its two endpoints, and for random modules
Corollary~\ref{cor:bridge_random} gives their exact expectations.

\paragraph{Experiment~6.} Two modules of eight vertices are joined by a single
interaction edge and observed with $G^{\le4}$. The modules are random digraphs
in which each ordered pair is an edge with probability $\bar d/7$, so that
the internal mean out-degree $\bar d$ grows from $0$ to $3$ in steps of $0.5$,
with
$1000$ realizations per step. The emergent edges rise steeply, from $0$ between
edgeless modules to $39.3\pm0.7$ at $\bar d=3$, and the bound rises from $1$ to
$111.2\pm3.4$. At every step, both agree with the exact expectations of
Corollary~\ref{cor:bridge_random} ($39.0$ and $110.8$ at $\bar d=3$) within their
$95\%$ confidence intervals (Figure~\ref{fig:exp_depend}(c)). The coupling is
the same single edge throughout; what grows is the internal structure it
connects. As the modules fill in, several boundary-crossing paths join the same
pair of vertices, and the emergent edges rise toward the $n_1n_2-1=63$ cross
pairs that one interaction edge can make emergent. The bridge law and its
equality condition also hold exactly on $3000$ random one-directional joins
($3$ to $12$ vertices per part, $|\mathcal I|\le4$, $L=2,3,4$). The expected
path and reach counts of Corollary~\ref{cor:bridge_random} agree exactly with
an enumeration of all digraphs on three, four, and five vertices.

\subsection{Scaling to large networks}\label{sec:scale}

\begin{figure}[tbp]
\centering
\includegraphics[width=\textwidth]{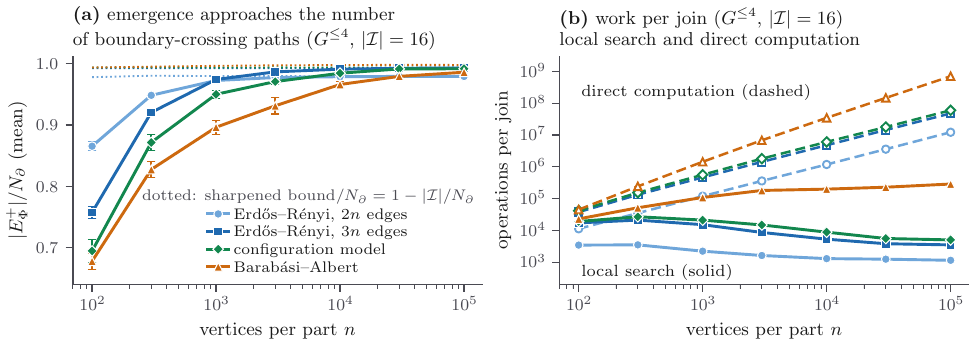}
\caption{\textbf{Experiment~7: the bound on large networks.} Two parts of $n$ vertices each, joined by $16$ random interaction edges from $s_1$ to $s_2$ and observed with $G^{\le4}$, for Erd\H{o}s--R\'enyi parts with $2n$ and $3n$ edges, heavy-tailed configuration parts, and Barab\'asi--Albert parts (from $50$ realizations per point for $n\le10^3$ to $20$ at $n=10^5$).
\textbf{(a)}~Mean of $|E^+_\Phi|/N_\partial$ with $95\%$ confidence intervals; dotted lines mark $1-|\mathcal I|/N_\partial$, the sharpened bound divided by $N_\partial$, since the interaction edges already lie in the observed union.
\textbf{(b)}~Mean number of operations per join: adjacency entries scanned by the local search of Corollary~\ref{cor:local} (solid), and breadth-first relaxations to depth $4$ from every vertex of $s_1$, $s_2$, and $s_1\vee s_2$ in the direct computation of $E^+_\Phi$ (dashed; estimated from $256$ sampled sources per graph).}
\alttext{Two panels. Left: the ratio of emergent edges to boundary-crossing paths, for four families of random networks, rises with the number of vertices per part from between 0.68 and 0.87 at one hundred vertices to between 0.98 and 0.99 at one hundred thousand, meeting the dotted lines of the sharpened bound near the top. Right: on logarithmic axes, the work of the direct computation grows in proportion to the number of vertices for three families and faster for the Barab\'asi--Albert family, from about ten thousand to between ten million and one billion operations, while the work of the local search stays between about one thousand and thirty thousand operations for three families and, for the Barab\'asi--Albert family, grows far more slowly than the direct computation.}
\label{fig:exp_scale}
\end{figure}

\paragraph{Experiment~7.} Parts of $n=10^2$ to $10^5$ vertices from four sparse
directed families (Erd\H{o}s--R\'enyi with $2n$ and $3n$
edges~\cite{erdos1959random}, and the Barab\'asi--Albert~\cite{barabasi1999emergence}
and heavy-tailed configuration~\cite{newman2001random} models of Experiment~2)
are joined by $1$, $4$, or $16$ random interaction edges from $s_1$ to $s_2$, or
$4$ or $16$ split between the two directions, and observed with $G^{\le L}$,
$L=2,3,4$ ($16{,}800$ instances). For graph powers, the local search also returns
the emergent edges, since membership of a produced pair in the observed union is
itself a local test (Appendix~\ref{app:computation}). The bound, the
sharpened bound, and $N_\partial\le\Delta_W$ (the walk bound of
Proposition~\ref{prop:complexity}) hold in every instance. On the
$14{,}588$ instances that we also computed directly, with up to $10^5$ vertices
per part, both computations give the same emergent edges.

The cost of the local search follows the neighborhoods of the interface
(Figure~\ref{fig:exp_scale}(b)). With $L=4$ and $16$ interaction edges from
$s_1$ to $s_2$, it needs no more operations at $n=10^5$ than at $n=10^2$ for the
Erd\H{o}s--R\'enyi and configuration parts (at most $5029$ per join on
average), while the work of the direct computation, which applies the
observation to the parts and to the join from scratch, grows in proportion to $n$,
by factors of $1084$ to $1435$. In the Barab\'asi--Albert parts, the local work
grows with the degrees of the hubs near the interface, from $2.3\times10^4$ to
$2.9\times10^5$ operations, as Corollary~\ref{cor:local} predicts, while the
direct work grows $1.6\times10^4$-fold. At $n=10^5$, the local search takes at
most $2.8$~ms per join, against $1.2$ to $7.4$~s for the direct computation
(Appendix~\ref{app:experiments}).

As $n$ grows, almost every boundary-crossing path other than an interaction
edge produces its own emergent edge
(Figure~\ref{fig:exp_scale}(a)). The interaction edges lie in the observed
union, so $|E^+_\Phi|/N_\partial\le1-|\mathcal I|/N_\partial$. With $L=4$ and
$16$ interaction edges from $s_1$ to $s_2$, the mean ratio rises from
$0.68$--$0.87$ at $n=10^2$ to $0.979$--$0.993$ at $n=10^5$; at $n=10^5$, for
every family and interface, the mean ratio lies within $0.02$ of
$1-|\mathcal I|/N_\partial$, the sharpened bound divided by $N_\partial$, at
$L=4$, and within $0.001$ at $L\le3$. In the
Erd\H{o}s--R\'enyi parts at $L=4$, the sharpened bound is attained in $183$ of
$200$ instances at $n=10^5$, against $23$ of $500$ at $n=10^2$: large sparse
networks are locally tree-like~\cite{newman2001random,newman2018networks}, and
their boundary-crossing paths account for emergence almost one for one.

\subsection{Comparison with simpler statistics}\label{sec:baselines}

We compare $N_\partial$, as a predictor of emergence, with four simpler
statistics of a coupling: the cut size $|\mathcal I|$; the \emph{degree sum}
$\sum_{(a,b)\in\mathcal I}(d^-(a)+d^+(b))$ and the \emph{degree product}
$\sum_{(a,b)\in\mathcal I}d^-(a)\,d^+(b)$, with the in-degree $d^-$ and the
out-degree $d^+$ taken in the parts; and the summed edge
betweenness~\cite{freeman1977set,girvan2002community} of the interaction edges
in $s_1\vee s_2$. At reach two, $N_\partial$ is itself a degree statistic,
$|\mathcal I|$ plus the degree sum plus the number of paths of length two
through two interaction edges, and a single coupling $(a,b)$ has
$|E^+_\Phi|=N_\partial-1=d^-(a)+d^+(b)$ under $G^{\le2}$. The comparison
therefore concerns observations that read three or more steps.

\paragraph{Experiment~8.} Parts of $20$ vertices from the thirteen synthetic
families of Experiment~2 at three densities each (mean total degree $1.9$ to
$8.0$), and random halves of its empirical networks, are joined by
$|\mathcal I|\in\{1,2,3,4\}$ random interaction edges, $30$ times per family,
density, and $|\mathcal I|$ ($90$ times per $|\mathcal I|$ for the empirical
networks; $5760$ instances). Under $G^{\le3}$, $G^{\le4}$,
and coarse-grained $G^{\le3}$, the Spearman correlation of $N_\partial$ with
$|E^+_\Phi|$ exceeds that of each of the four statistics, and every paired
difference has a bootstrap $95\%$ interval above zero
(Table~\ref{tab:baselines}). The advantage persists within the $168$ cells of
fixed family, density, and $|\mathcal I|$, in each of the $16$ families under
$G^{\le3}$, and for parts of $12$ and $30$ vertices. To rank candidate
couplings, we score candidate single interaction edges $(a,b)\in V_1\times V_2$
for $480$ pairs of parts ($178{,}350$ candidates). The candidate with the
largest $N_\partial$ has
maximum emergence in $80\%$ of the pairs under $G^{\le3}$ and $55\%$ under
$G^{\le4}$, against at most $63\%$ and $37\%$ for the other statistics, and
realizes on average $99\%$ and $96\%$ of the largest emergence available.

\begin{table}[tbp]
\centering
\caption{\textbf{Experiment~8: predicting and ranking emergence.} Left:
Spearman correlation with $|E^+_\Phi|$ over $5760$ couplings under $G^{\le3}$,
$G^{\le4}$, and coarse-grained (c.-g.) $G^{\le3}$ (bootstrap $95\%$
intervals within $\pm0.03$). Right: fraction of $480$ pairs of parts in which
the candidate single coupling that scores highest has maximum $|E^+_\Phi|$
(ties in the score split evenly; intervals within $\pm0.05$); for a single
coupling the cut size ranks at random.}
\label{tab:baselines}
\footnotesize
\setlength{\tabcolsep}{4pt}
\begin{tabular}{@{}lrrr@{\hspace{14pt}}rrr@{}}
\toprule
& \multicolumn{3}{c}{rank correlation} & \multicolumn{3}{c}{best coupling found}\\
\cmidrule(lr){2-4}\cmidrule(lr){5-7}
statistic & $G^{\le3}$ & $G^{\le4}$ & c.-g.\ $G^{\le3}$ & $G^{\le3}$ & $G^{\le4}$ & c.-g.\ $G^{\le3}$\\
\midrule
$N_\partial$            & 0.989 & 0.971 & 0.932 & 0.80 & 0.55 & 0.40\\
degree product          & 0.958 & 0.943 & 0.895 & 0.63 & 0.37 & 0.37\\
degree sum              & 0.950 & 0.891 & 0.888 & 0.62 & 0.37 & 0.37\\
edge betweenness        & 0.841 & 0.894 & 0.822 & 0.05 & 0.06 & 0.04\\
cut size $|\mathcal I|$ & 0.509 & 0.417 & 0.494 & 0.01 & 0.01 & 0.02\\
\bottomrule
\end{tabular}
\end{table}

\subsection{Real directed networks: attribution and intervention}\label{sec:real}

\paragraph{Experiment~9.} We apply the bound and the shares to three directed
networks whose parts are recorded in the data (Table~\ref{tab:real} in
Appendix~\ref{app:experiments}): the chemical-synapse connectome of the
\emph{C.~elegans} hermaphrodite~\cite{cook2019whole}, with its $83$ sensory
neurons, $81$ interneurons, and $108$ motor neurons as parts; the e-mail network
of a European research institution~\cite{leskovec2007graph,yin2017local}, with
its $42$ departments as parts (Remark~\ref{rem:multipart}); and the hyperlinks
among U.S.\ political blogs~\cite{adamic2005political}, with the liberal and the
conservative blogs as parts. The interaction is the set of edges between parts,
and the observations are $G^{\le2}$, $G^{\le3}$, and $G^{\le2}$ coarse-grained
by blocks from the data (neuron classes~\cite{white1986structure}, departments).
The emergent edges recovered by Theorem~\ref{thm:main}(i) from the paths listed
by the local search coincide with those of a direct computation. For the e-mail
network under $G^{\le3}$, its $8.4\times10^{7}$ boundary-crossing paths are
counted by the path-count identity of Corollary~\ref{cor:pathcount}.

\begin{figure}[tbp]
\centering
\includegraphics[width=\textwidth]{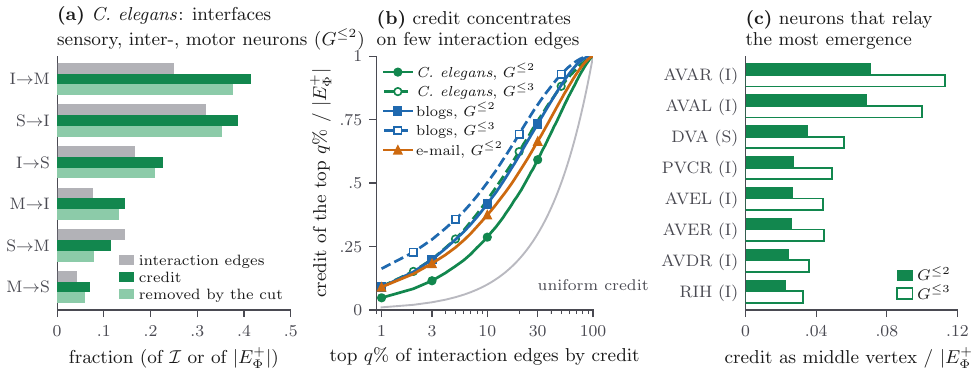}
\caption{\textbf{Experiment~9: attribution in real directed networks.}
\textbf{(a)}~\emph{C.~elegans} connectome under $G^{\le2}$: for each interface (ordered pair of neuron types: S, sensory; I, interneuron; M, motor), its fraction of the interaction edges (gray), its credit as a fraction of $|E^+_\Phi|$ (dark), and the fraction of emergent edges removed by cutting it (light), as predicted by Proposition~\ref{prop:intervene} and confirmed by recomputation; the interface credits sum to more than $|E^+_\Phi|$ because a route can cross two interfaces.
\textbf{(b)}~Credit of the top $q\%$ of interaction edges, taken as a set, as a fraction of $|E^+_\Phi|$ (logarithmic $q$ axis; gray, uniform credit).
\textbf{(c)}~Credit as a middle vertex, the total share of the producing routes that pass through a neuron as a middle vertex, as a fraction of $|E^+_\Phi|$, for the eight neurons that relay the most emergence under $G^{\le2}$ (filled bars; open bars, under $G^{\le3}$).}
\alttext{Three panels. Left: grouped horizontal bars for the six interfaces between sensory neurons, interneurons, and motor neurons; the interneuron-to-motor and sensory-to-interneuron interfaces carry the largest credit, about 0.42 and 0.39 of the emergent edges, more than their fractions of the interaction edges, and cutting them removes slightly less than their credit. Middle: five rising concentration curves for the connectome, the blog network, and the e-mail network, all far above the curve of uniform credit; the top ten percent of interaction edges carry between about 0.29 and 0.50 of the emergence. Right: paired horizontal bars for eight neurons, led by the interneurons AVAR and AVAL, with larger credit as a middle vertex at reach three than at reach two.}
\label{fig:exp_real}
\end{figure}

Emergence is abundant: under $G^{\le2}$, the connectome gains $16{,}944$
emergent edges and the blog network $37{,}290$, from $42{,}682$ and $89{,}544$
boundary-crossing paths. In these clustered, densely connected networks, an
emergent edge is produced by more than one route on average: $1.8$ in the
connectome
and $1.9$ in the blog network under $G^{\le2}$, and $13.2$ and $16.0$ under
$G^{\le3}$; $N_\partial$ counts all of them, and the shares divide each emergent
edge among them. In the connectome, $251$ of the $8964$ ordered pairs of a
sensory and a motor neuron are joined by a synapse; $G^{\le2}$ adds $3283$
sensory-to-motor edges, $1963$ of them produced only through interneurons, and
$G^{\le3}$ adds $7612$, $4265$ of them produced only through interneurons.
As middle vertices of the producing routes, interneurons relay $66\%$ of the
emergence under $G^{\le2}$. The leading relays are the interneurons AVAR and
AVAL ($7.1\%$ and $6.9\%$), followed by DVA (typed as sensory), PVCR, AVEL,
AVER, and AVDR (Figure~\ref{fig:exp_real}(c)); AVDR and PVCR belong to the
classes AVD and PVC, which lie on the pathways from touch receptors to
movement~\cite{chalfie1985neural}. The synapses from
interneurons to motor neurons are $25\%$ of the interaction edges but carry
credit $0.42\,|E^+_\Phi|$; cutting them removes $38\%$ of the emergent
edges ($6371$), while $107$ images of the removed synapses become emergent, as
Proposition~\ref{prop:intervene}(iii) states. The direct sensory-to-motor synapses are $15\%$ of the interaction
edges and carry $0.11\,|E^+_\Phi|$ (Figure~\ref{fig:exp_real}(a)).

Credit concentrates on few interaction edges (Figure~\ref{fig:exp_real}(b)).
Under $G^{\le3}$, the top $10\%$ of interaction edges carry credit
$0.43\,|E^+_\Phi|$ in the connectome and $0.50\,|E^+_\Phi|$ in the blog
network, where the ten links of largest credit, $0.6\%$ of the $1683$ links
between the camps, carry $0.12\,|E^+_\Phi|$, and a single link is necessary for
$4598$ emergent edges: every route that produces them uses it. Among departments, $387$ ordered pairs without e-mail
between them are linked through a member of a third department, and one
department relays $40\%$ of them. Removing interaction edges changes emergence
exactly as Proposition~\ref{prop:intervene} states. In all $42$
interventions (a whole interface, the ten edges of largest credit, or the edge
necessary for the most emergent edges), the emergent edges that vanish are
exactly those all of whose producing routes use the removed edges; the new ones
are exactly the images of removed edges, absent from the reduced observed
union, that other routes still produce; and the loss never exceeds the credit
removed. The properties of the shares and Proposition~\ref{prop:intervene}
also hold exactly on $3200$ random joins of two or three parts, with all
$20{,}872$ removals of sets of interaction edges
(Appendix~\ref{app:experiments}).

\subsection{Summary}\label{sec:exp_summary}

The identity and the bound hold exactly in every experiment, on networks of up
to $10^5$ vertices per part. For observations acting on paths, emergence requires routes of length at
least two and attains the sharpened bound where the routes through the interface are
distinct. It rises with the reach and granularity of the observation, depends, at
mean total degree $4$, on topology once the reach is three or more, and surges, as the bridge law
predicts, when one edge bridges internally structured modules. The local search
computes the bound at a cost set by the interface, and, for observations that
read three or more steps, $N_\partial$ predicts and ranks couplings better than
the cut size, the degrees at the interface, and edge betweenness. On real networks, the shares locate the interfaces, edges,
and relay vertices that carry emergence, and
Proposition~\ref{prop:intervene} predicts exactly what removing interaction
edges does.

\section{Discussion}\label{sec:discussion}

\subsection{Relation to existing work}\label{sec:related}

\paragraph{Generative effects.}  For systems combined by a semilattice join,
Adam and Dahleh~\cite{adam2019generativity}, following Adam's
thesis~\cite{adam2017thesis}, say that an observation $\Phi$
sustains generative effects when $\Phi(s\vee s')\neq\Phi(s)\vee\Phi(s')$ for
some systems $s,s'$, with reachability in a union of directed graphs as one
example.  They ask how to express $\Phi(s\vee s')$ through the parts.  Our
discrepancy adapts this inequality to graphs coupled by an explicit
interaction and measures it by the edges gained and lost.  For observations
acting on paths, Theorem~\ref{thm:main} gives an answer: the whole adds
exactly the new edges produced by boundary-crossing paths.  Our prior work
counted the paths that start at the heads of the edges an observation
deletes~\cite{li2025categorical}; here the counted paths cross the
interaction, and they bound and attribute emergence.

\paragraph{Measures of emergence.}  Causal emergence compares the
effective information of micro- and macro-level
descriptions~\cite{hoel2013quantifying,hoel2017map}, including those of networks
whose nodes are grouped into macro-nodes~\cite{klein2020emergence}; information
decomposition detects synergy of a whole beyond its parts in multivariate
data~\cite{rosas2020reconciling,mediano2022greater,varley2022emergence};
integrated information measures how far a whole exceeds its parts across its
weakest partition~\cite{tononi2004information,oizumi2014phenomenology}; and
closure criteria organize emergent levels into
hierarchies~\cite{rosas2024software}.  These measures are
computed from dynamics or statistics, and many search over partitions or
coarse-grainings.  Ours is computed from structure, for given parts and a
given observation, and assigns each unit of emergence to routes.  Weak emergence
holds that macro properties can be derived from the micro level only by
working through its interactions~\cite{bedau1997weak}; for observations acting
on paths, the main theorem names the routes of such a derivation.

\paragraph{Interacting and multilayer networks.}  Multilayer and
interdependent networks make the coupling between parts
explicit~\cite{kivela2014multilayer,dedomenico2013mathematical,boccaletti2014structure,gao2012networks}
and show that it transforms system-level
behavior~\cite{radicchi2013abrupt,gomez2013diffusion,cardillo2013emergence}:
interdependence produces abrupt cascades of
failures~\cite{buldyrev2010catastrophic} that tuning the interconnection can
suppress or amplify~\cite{brummitt2012suppressing}, as cutting boundary-crossing paths and
bridging modules do for emergence here (Propositions~\ref{prop:cut}
and~\ref{prop:bridge}).  That literature studies how coupling changes
processes and global descriptors; we bound what an observation of the coupled
structure reveals by the routes through the coupling.

\paragraph{Paths, walks, and flows.}  Many network measures are built from
the routes along which something flows~\cite{borgatti2005centrality}.  Katz
centrality counts attenuated walks~\cite{katz1953new}, betweenness counts the
shortest paths that pass through a vertex or an
edge~\cite{freeman1977set,girvan2002community}, structural-hole theory ascribes
an advantage to a vertex that joins otherwise unconnected
contacts~\cite{burt1992structural}, and communicability and global efficiency
aggregate walks and shortest paths, respectively, over vertex
pairs~\cite{estrada2008communicability,latora2001efficient}.  These tools
describe a single network.  We relate a whole to its parts: the
boundary-crossing paths account for the structure found only in the whole,
exactly for observables that sum over paths (Theorem~\ref{thm:scalar}), and
predict it better than edge betweenness or the degrees at the interface once
the observation reads three or more steps (Section~\ref{sec:baselines}).

\paragraph{Incremental computation.}  Accounting derivation by derivation for
what an insertion adds is classical for single queries.  Delta rules find the
new answers of a database view among the derivations that use an inserted
fact, and a counting algorithm tracks how many derivations each answer
has~\cite{gupta1993maintaining}.  Incremental algorithms maintain the
transitive closure of a directed graph under edge
insertions~\cite{ibaraki1983online,italiano1986amortized}, where inserting
$(a,b)$ joins each vertex reaching $a$ to each vertex reached from $b$, the
pairs that Proposition~\ref{prop:bridge} counts within reach $L$.  Provenance
semirings sum over derivations~\cite{green2007provenance} as
Theorem~\ref{thm:scalar} sums over paths.  Lemma~\ref{lem:crossing} is, in
these terms, the delta rule for inserting the interaction edges into a view of
paths of bounded length; we use it to bound and attribute emergence by
boundary-crossing paths for a whole class of observations.

\paragraph{Coarse-graining and renormalization.}  Network renormalization
(see~\cite{gabrielli2025network} for a review) coarse-grains a network across
scales and compares it with its own coarse versions.  The coarse-graining
proceeds by box covering~\cite{song2005selfsimilarity}, by preserving the slow
modes of random walks~\cite{gfeller2007spectral}, through a hidden metric
geometry~\cite{garciaperez2018multiscale}, or by diffusion along the routes by
which signals spread~\cite{villegas2023laplacian}, the scheme closest in
spirit to ours.  We compare a whole with its parts at a fixed scale:
coarse-graining by a fixed partition is one of our observations, and, composed
with a graph power, it produces emergence that the boundary-crossing paths
bound (Example~\ref{ex:coarse}, Section~\ref{sec:granularity}).

\paragraph{Graph operators and cuts.}  Operators that send graphs to graphs,
such as powers and line graphs, are classical objects of
study~\cite{prisner1995graph}, and, for undirected graphs, cuts and their
expansion are central to spectral graph theory~\cite{chung1997spectral}.  We
ask how an operator interacts with joining two graphs across a cut: for the
observations acting on paths (Definition~\ref{def:path_obs}), the answer
counts the routes through the cut rather than its edges, and a single edge can
carry many routes (Example~\ref{ex:bowtie}).

\paragraph{Emergence and circuits in learned systems.}  Large language models
show abilities absent at smaller scales and hard to
predict~\cite{wei2022emergent,ganguli2022predictability}, and apparent
emergence depends on the metric~\cite{schaeffer2023emergent}, as emergence
here depends on the observation.  In-context learning improves abruptly as induction heads, a
circuit of attention heads in different layers,
form~\cite{olsson2022incontext}: a capability emerges as a route forms across
components.  Mechanistic interpretability attributes behavior to circuits,
subgraphs of the computational graph found by inspection or
intervention~\cite{olah2020zoom,wang2023interpretability,conmy2023automated},
and path expansions decompose residual networks and attention-only
transformers into the contributions of individual
routes~\cite{veit2016residual,elhage2021mathematical}.  The shares of
Definition~\ref{def:shares} are a graph-theoretic counterpart of this
attribution: they divide each emergent feature among the routes across a
chosen interface that produce it.  The number of boundary-crossing paths that
produce an edge bounds the emergence (Theorem~\ref{thm:main}), and
Proposition~\ref{prop:intervene} gives the exact effect of removing
interaction edges.

\subsection{Applications}\label{sec:applications}

\paragraph{Brain networks.}  Brain function arises from structural
connectivity and interactions among
modules~\cite{bullmore2009complex,park2013structural}, and recording
modalities resolve this structure at scales from single units to
regions~\cite{bassett2017network}.
Each modality can be read as an observation of one structural network, such as
a coarse-graining into regions composed with a graph power whose reach
reflects multistep signaling.  Typing neurons by function
makes a connectome a system of interacting parts: in \emph{C.~elegans}, the
credits of interfaces and relay neurons rank the synapses and interneurons
through which sensory-to-motor reach arises, and
Proposition~\ref{prop:intervene} states which part of it is abolished when
synapses are removed (Section~\ref{sec:real}).

\paragraph{Coupling modules in engineered and biological systems.}
Engineering a modular system, such as a power grid coupled to the
communication network that controls it~\cite{buldyrev2010catastrophic} or a
gene regulatory circuit~\cite{alon2007network}, involves choosing where to add
or cut interaction edges.  Because $N_\partial$ is computed locally, candidate
couplings can be ranked by the boundary-crossing paths they would open before
any is made, as in Experiment~8 (Section~\ref{sec:baselines}).  For a
candidate cut, Proposition~\ref{prop:intervene} names the emergent edges it
destroys, those all of whose producing routes use the removed edges, and bounds
their number by the credit of the removed edges.

\paragraph{Interface-level attribution in computational graphs.}  In a
feedforward network, a directed acyclic graph, blocks of layers define parts
and interfaces.  Under an observation acting on paths, such as a
coarse-grained power, the shares identify the routes across each interface
responsible for each emergent feature, Proposition~\ref{prop:intervene}
predicts the effect of removing connections between blocks, and the number
$|\mathcal P_\partial|$ of boundary-crossing paths, which bounds $N_\partial$
and the total discrepancy, is computed exactly by sparse matrix--vector
products (Appendix~\ref{app:computation}).

\subsection{Open questions}\label{sec:open}

\paragraph{Edges confirmed by an alternative route.}  Some observations keep
an edge $(u,v)$ when another short route also joins $u$ to $v$.  Extending the witness map from single paths to
such edge--route pairs would bring these observations within the bound.

\paragraph{Partitions computed from the graph.}  Extending the framework from
partitions fixed in advance to partitions computed from the graph, such as
communities~\cite{fortunato2010community} or strongly connected
components~\cite{bang2009digraphs}, whose vertex map changes when the parts
are joined, would cover these widely used observations.  A search for the
interface that carries the most boundary-crossing paths would be a structural
counterpart of the partition search of integrated
information~\cite{oizumi2014phenomenology}.

\paragraph{Dynamics on networks.}  Letting the network and the observation
evolve in time, or observing a process on the network, such as the Boolean
dynamics of gene regulation~\cite{kauffman1993origins}, would connect the
structural bound to dynamical emergence and to the information-theoretic
measures above.

\paragraph{Overlapping parts.}  The scalar identity (Theorem~\ref{thm:scalar})
accounts for the paths shared by overlapping parts; extending the witness map
to such parts, and to interactions coupling several parts at once, would cover
overlapping communities~\cite{fortunato2010community} and higher-order
structures.

\paragraph{Emergence from summary statistics.}  The bridge law gives the
expected emergence of a single bridge between random modules exactly
(Corollary~\ref{cor:bridge_random}).  Extending such expectations to many
interaction edges and to random graphs with heterogeneous degrees, and
relating $N_\partial$ to the expansion of the cut through Cheeger-type
inequalities for the underlying undirected graph~\cite{chung1997spectral},
would yield estimates of emergence from summary statistics and connect the
bound to the walk counts of Appendix~\ref{app:potential}.

\section{Conclusion}\label{sec:conclusion}

We have proposed a structural view of emergence: a system is a network of
interacting parts, an observation is a way of looking at it, and emergence is
what observing the coupled whole reveals beyond the observed parts and their
interaction. For observations that act on paths, the emergent edges are
exactly the new edges produced by boundary-crossing paths, and distinct
emergent edges have distinct witnesses. The number $N_\partial$ of
boundary-crossing paths that produce an edge therefore bounds the number of
emergent edges, with an exact criterion for equality. The
number of all boundary-crossing paths, producing or not, bounds the gained and
lost edges together, and these paths account exactly for the emergence of
observables that sum over paths.

To our knowledge, the result is new: no previous work bounds and attributes
emergence by the routes that cross the interface between parts. It is
mechanistic: it anticipates emergence and shows how to suppress it by cutting
boundary-crossing paths or amplify it by bridging internally structured modules, where a single
bridge pairs the in-reach of its tail with the out-reach of its head, with
exact expectations for random modules. It attributes each emergent feature to
the components and interactions along the routes that produce it, and predicts
exactly how emergence changes when interaction edges are removed. It is
computed locally around the interface, in milliseconds on sparse parts of $10^5$
vertices, and it is grounded in the fact that a network and its paths carry
the same information. Exact computations on $8935$ instances confirm the
identity and the bound in every case. Other experiments show that at mean
total degree $4$ the influence of topology grows with reach and that emergence
rises with coarse-graining granularity and surges when one edge bridges
internally structured modules. They also show that, for observations reading
three or more steps, the boundary-crossing paths predict emergence better than
the cut size, the degrees at the interface, and edge betweenness, and that in
real directed networks these paths locate emergence at specific interfaces and
relay vertices, such as the interneurons of \emph{C.~elegans}. Extending the
bound to partitions computed from the network and to dynamics would carry
a simple message from structure to function: emergence travels along the
routes a coupling opens, and counting them tells how much to expect and where
it comes from.

\appendix
\counterwithin{table}{section}  
\section{Computing the discrepancy and the bound}\label{app:computation}

\paragraph{Direct computation.}  The total discrepancy $|E^+_\Phi|+|E^-_\Phi|$
is computed by applying $\Phi$ to $s_1$, $s_2$, and $s_1\vee s_2$, adding
$\varphi(\mathcal I)$ to the union of the outputs for the parts, and comparing
edge sets with a hash table, in time linear in $|\mathcal I|$, in the number of
edges of the three outputs, and in the time of the three applications.

\paragraph{Local search.}  For each interaction edge $(u,v)$, the search of
Corollary~\ref{cor:local}(ii) runs a depth-first search backward from $u$ and
forward from $v$ and combines pieces of combined length at most $L-1$; each
piece is obtained from a shorter one in one step, and each combination that
forms a path, or a cycle for observations that read cycles (each rotation of a
cycle arises as its own combination), is tested for production, a test counting
as one step.  To list each
boundary-crossing path once, the search assigns it to the first interaction
edge $(u,v)$ it traverses: the portion before that edge lies in the part $s$
that contains $u$ and is found by a backward search there, the rest by a
forward search in $s_1\vee s_2$, and every combination is a path unless the
forward search re-enters $s$.  Building a witness map also requires deciding
which produced edges lie in $E(\Phi(s_1)\vee\Phi(s_2))$, which may involve
applying $\Phi$ to the parts; for graph powers this test is itself local, since
$(x,y)\in E(\Phi(s_1)\vee\Phi(s_2))$ exactly when $(x,y)\in\mathcal I$ or a path
of length at most $L$ inside one part joins $x$ to $y$.  For graph powers, the
forward routes that stay in the part $s'$ of $v$ produce exactly the pairs
$(x,y)$ with $\mathrm{dist}_s(x,u)+\mathrm{dist}_{s'}(v,y)\le L-1$, which two
bounded breadth-first searches list; Experiment~7 computes the emergent edges
in this way on parts of up to $10^5$ vertices (Section~\ref{sec:scale}).
Counting the simple paths between two vertices is \#P-complete in
general~\cite{valiant1979complexity}; the bound needs only the paths of length
at most $L$ through the interaction edges, which the local search lists at a
cost independent of the size of the parts.

\paragraph{Walk bound.}  Write $|E|$ for the number of edges of $s_1\vee s_2$,
$A_G$ for the adjacency matrix of a graph $G$, and
\[
W_L(G):=\sum_{k=1}^{L}\mathbf 1^{\top}A_G^{\,k}\,\mathbf 1
\]
for the number of walks of length $1$ to $L$ in $G$, since the entry $(u,v)$ of
$A_G^{\,k}$ counts the walks of length $k$ from $u$ to
$v$~\cite{biggs1993algebraic,newman2018networks}.

\begin{proposition}[Walk bound]\label{prop:complexity}
In the setting of Theorem~\ref{thm:main}, the walk discrepancy
$\Delta_W:=W_L(s_1\vee s_2)-W_L(s_1)-W_L(s_2)$ is the number of walks of length
at most $L$ in $s_1\vee s_2$ that traverse an interaction edge, so
$\Delta_W\ge|\mathcal P_\partial|\ge N_\partial\ge|E^+_\Phi|$ and
$|\mathcal P_\partial|\ge|E^+_\Phi|+|E^-_\Phi|$.  It is computed by $3L$ sparse
matrix--vector products, in time $O(|E|\,L)$ when every vertex has an edge, and
$\Delta_W=|\mathcal P_\partial|$ when $s_1\vee s_2$ has no directed cycle of
length at most $L$, for instance when it is acyclic.
\end{proposition}

\begin{proof}
The proof of Lemma~\ref{lem:crossing} never uses that the vertices of a path
are distinct, so it applies to walks: the walks of $s_1\vee s_2$ are those of $s_1$,
those of $s_2$, and those that traverse an interaction edge, and no walk is of
two of these kinds.  This gives the first claim, and the inequalities follow
because every path or cycle of $\mathcal P_\partial$ is such a walk and
distinct ones are distinct walks (Theorem~\ref{thm:main} and
Corollary~\ref{cor:discrepancy}).  For each of the three graphs,
$W_L(G)=\sum_{k=1}^{L}\mathbf 1^{\top}x_k$ with $x_0=\mathbf 1$ and
$x_k=A_Gx_{k-1}$, which takes $L$ sparse matrix--vector products of cost
$O(|E|)$ each.  A walk of length at most $L$ that repeats a vertex contains a
directed cycle of length at most $L$, so without such cycles every walk is a
path.
\end{proof}

Computing $\Delta_W$ is cheaper than the local search when
$|\mathcal I|\,d^{\,L-1}$ exceeds $|E|$, with $d$ as in
Corollary~\ref{cor:local}, as for dense interfaces, high-degree boundary
vertices, or long reach, since its cost grows linearly in $L$.  On acyclic
systems, such as feedforward networks, $\Delta_W$ is the path-count discrepancy
of Corollary~\ref{cor:pathcount}, which equals $N_\partial$ for $G^{\le L}$.

\section{The emergence potential of a subsystem}\label{app:potential}

The partner of a subsystem is often unknown.  One then asks how much emergence a
subsystem $s$ of a larger system $G$ can produce with any partner that $G$
provides.  The walks through the boundary of $s$ answer this question.
Throughout, $G$ is a finite directed graph with adjacency matrix $A_G$,
$s\subseteq G$ is a subgraph, and $\Delta_{\#}$ denotes the scalar discrepancy
(Definition~\ref{def:discrepancy_scalar}) of the path count
$H\mapsto|\mathcal P_L(H)|$, the number of paths of length at most $L$.

\begin{definition}[Emergence potential]\label{def:potential}
The \emph{boundary} of $s$ in $G$ is the set $\partial(s)$ of vertices of $s$
joined by an edge of $G$ to a vertex outside $V(s)$.  A \emph{coupling} of $s$ in
$G$ is a subgraph $s'\subseteq G$ with $V(s')\cap V(s)=\varnothing$ together with
an interaction $\mathcal I$ consisting of edges of $G$ between $V(s)$ and $V(s')$.
The \emph{emergence potential} of $s$ in $G$ is
\[
\mathcal E^L(s;G):=\max_{(s',\mathcal I)}\Delta_{\#}(s,s';\mathcal I),
\]
the maximum over all couplings of $s$ in $G$.
\end{definition}

For every coupling, the endpoints in $V(s)$ of the interaction edges lie in
$\partial(s)$.  By Corollary~\ref{cor:pathcount},
$\Delta_{\#}(s,s';\mathcal I)=|\mathcal P_\partial|$ is the number of
boundary-crossing paths of the coupling.  Adding edges to $s'$ or to
$\mathcal I$ only adds such paths, so the maximum is attained by the subgraph of
$G$ induced on $V(G)\setminus V(s)$, coupled to $s$ through every edge of $G$
between them.  The potential thus measures the couplings that $G$ actually
provides.

\begin{definition}[Walks through a set]\label{def:through_boundary}
For a vertex $v$ of $G$, let $\mathrm{in}_a(v):=(\mathbf 1^{\top}A_G^{\,a})_v$ and
$\mathrm{out}_b(v):=(A_G^{\,b}\mathbf 1)_v$ be the numbers of walks of length $a$
ending at $v$ and of length $b$ starting at $v$, with
$\mathrm{in}_0(v)=\mathrm{out}_0(v)=1$.  For $S\subseteq V(G)$, set
\[
W^L_{\ni}(S;G):=\sum_{v\in S}\ \sum_{k=1}^{L}\ \sum_{a=0}^{k}
\mathrm{in}_a(v)\,\mathrm{out}_{k-a}(v).
\]
\end{definition}

\begin{figure}[htbp]
\centering
\begin{tikzpicture}[x=1mm,y=1mm,
    brc/.style = {decorate,decoration={brace,amplitude=5pt},thick,draw=obscol},
    bl/.style  = {font=\footnotesize,align=center,above=6pt},
  ]
  \node[v1,vb] (v)  at (0,0)    {$v$};
  \node[vn]    (x1) at (-34,8)  {};
  \node[vn]    (y1) at (-17,8)  {};
  \node[vn]    (x2) at (-34,-8) {};
  \node[vn]    (y2) at (-17,-8) {};
  \node[vn]    (z1) at (17,8)   {};
  \node[vn]    (w1) at (34,8)   {};
  \node[vn]    (z2) at (17,-8)  {};
  \node[vn]    (w2) at (34,-8)  {};
  \draw[eobs] (x1) -- (y1); \draw[eobs] (y1) -- (v);
  \draw[eobs] (x2) -- (y2); \draw[eobs] (y2) -- (v);
  \draw[eobs] (v) -- (z1);  \draw[eobs] (z1) -- (w1);
  \draw[eobs] (v) -- (z2);  \draw[eobs] (z2) -- (w2);
  \begin{scope}[on background layer]
    \draw[obscol!45,densely dotted,thick] (0,-15) -- (0,15);
    \draw[draw=gcol,line width=7pt,line cap=round,line join=round,opacity=.3]
      (x1.center) -- (y1.center) -- (v.center) -- (z2.center) -- (w2.center);
  \end{scope}
  \draw[brc] (-37.25,13.5) -- (-4,13.5)
    node[bl,midway] {in-walks of length $a$\\
                     $\mathrm{in}_a(v)=(\mathbf{1}^{\top}A_G^{\,a})_v$};
  \draw[brc] (4,13.5) -- (37.25,13.5)
    node[bl,midway] {out-walks of length $k-a$\\
                     $\mathrm{out}_{k-a}(v)=(A_G^{\,k-a}\mathbf{1})_v$};
  \node[font=\scriptsize,text=obscol] at (0,-17.5) {cut at position $a$};
  \node[font=\footnotesize] at (0,-24)
    {number of walks of length $k$ with $v$ at position $a$
     $\;=\;\mathrm{in}_a(v)\times\mathrm{out}_{k-a}(v)$};
\end{tikzpicture}
\caption{\textbf{A walk through a boundary vertex splits into an in-walk
and an out-walk (Definition~\ref{def:through_boundary} and the proof of
Theorem~\ref{thm:potential}).}  Every such pair glues back into one walk,
so here two in-walks and two out-walks of length two give four walks of length
four with $v$ in the middle (one highlighted).}
\alttext{A double-ringed blue vertex v sits in the middle.  Two chains of
two gray vertices lead into v from the upper left and the lower left, and
two chains lead out of v to the upper right and the lower right.  A green
band highlights one walk that enters from the upper left and leaves to
the lower right, and a dotted vertical line through v, labeled ``cut at
position a,'' separates the two sides.  Braces label the left side as
in-walks of length a ending at v and the right side as out-walks of
length k minus a starting at v, each with a formula for its count, and a
line below states that the number of walks of length k with v at
position a is the product of the two counts.}
\label{fig:throughboundary}
\end{figure}

\begin{theorem}[Emergence potential is bounded by walks through the boundary]\label{thm:potential}
For every subgraph $s\subseteq G$,
\[
\mathcal E^L(s;G)\;\le\;W^L_{\ni}\bigl(\partial(s);G\bigr).
\]
Moreover, for every observation $\Phi$ acting on paths of length at most $L$ and
every coupling of $s$ in $G$, $|E^+_\Phi|\le N_\partial\le W^L_{\ni}(\partial(s);G)$.
\end{theorem}

\begin{proof}
Fix a coupling $(s',\mathcal I)$.  Its join is a subgraph of $G$, since
$s,s'\subseteq G$ and $\mathcal I\subseteq E(G)$, and a path or cycle of
$\mathcal P_\partial$ traverses an interaction edge, whose endpoint in $V(s)$
lies in $\partial(s)$; so it is a walk of $G$ of length at most $L$ that visits
$\partial(s)$, and distinct ones are distinct walks.  Cutting a walk
$(\omega_0,\dots,\omega_k)$ with $\omega_a=v$ at position $a$ gives a walk of
length $a$ ending at $v$ and one of length $k-a$ starting at $v$, and every
such pair glues back into one walk (Figure~\ref{fig:throughboundary}); so
$\mathrm{in}_a(v)\,\mathrm{out}_{k-a}(v)$ counts the walks of length $k$ with $v$
at position $a$, and $W^L_{\ni}(\partial(s);G)$ counts each walk of length $1$
to $L$ once for every position at which it visits $\partial(s)$.  Hence
$|\mathcal P_\partial|\le W^L_{\ni}(\partial(s);G)$, with or without cycles;
for the path count, $|\mathcal P_\partial|=\Delta_{\#}(s,s';\mathcal I)$
(Corollary~\ref{cor:pathcount}), and for $\Phi$,
$|E^+_\Phi|\le N_\partial\le|\mathcal P_\partial|$ (Theorem~\ref{thm:main}).  The right-hand side does not depend on the
coupling, so it also bounds the maximum.
\end{proof}

The capacity of a subsystem for emergence is thus carried by its boundary
vertices and by the reach of $G$ around them, and the bound decomposes into
classical walk centralities.  Let
$\kappa^{\mathrm{out}}_L(v):=\sum_{k=1}^{L}(A_G^{\,k}\mathbf 1)_v$ and
$\kappa^{\mathrm{in}}_L(v):=\sum_{k=1}^{L}(\mathbf 1^{\top}A_G^{\,k})_v$ be the
truncated Katz out- and in-centralities of $v$ at unit
attenuation~\cite{katz1953new}.  Splitting the inner sum of
Definition~\ref{def:through_boundary} into the terms with $a=0$, $a=k$, and
$0<a<k$, with
$\mathrm{in}_0=\mathrm{out}_0=1$, gives
\[
W^L_{\ni}\bigl(\partial(s);G\bigr)=\sum_{v\in\partial(s)}\Bigl[\kappa^{\mathrm{out}}_L(v)
+\kappa^{\mathrm{in}}_L(v)+\sum_{k=2}^{L}\sum_{a=1}^{k-1}
\mathrm{in}_a(v)\,\mathrm{out}_{k-a}(v)\Bigr]:
\]
the outgoing part of the bound, the terms with $a=0$, is the truncated Katz
out-centrality summed over the boundary.

The three terms have a direct reading: how far the boundary vertices reach
forward, how far they are reached, and how many walks pass through them.  For
the full Katz and exponential walk sums, rankings approach degree
centrality as the attenuation parameter tends to zero and eigenvector
centrality as it tends to its upper limit~\cite{benzi2015limiting}.  The truncated sums
used here vary similarly with $L$: at $L=1$ they are degrees, and when $G$ is
strongly connected and aperiodic, $\kappa^{\mathrm{out}}_L$ and
$\kappa^{\mathrm{in}}_L$, normalized, converge to the right and left Perron
eigenvectors of $A_G$ as $L$ grows, by a power-iteration argument.  The
potential of a subsystem is thus governed by well-understood properties of its
boundary.  All the quantities are obtained from the $2L$ vectors
$\mathbf 1^{\top}A_G^{\,a}$ and $A_G^{\,b}\mathbf 1$ with $1\le a,b\le L$, so ranking
the subsystems of a large system by $W^L_{\ni}(\partial(s);G)$ costs $2L$ sparse
matrix--vector products followed by $O(L^2)$ operations per boundary vertex.

\section{Details of the experiments}\label{app:experiments}

\paragraph{Path rules.} The path rules of drop sinks, edges on a cycle, the
continuation filter with $L=3$, and coarse-grained $G^{\le2}$ reproduce the
directly defined observations on all $400$ random graphs tested.  Those of
drop sinks and edges on a cycle (of length at most $2$, $3$, $4$, or any length)
also reproduce them on $1200$ further random graphs of $5$ to $10$ vertices and
on every part and join of Experiment~1.

\paragraph{Experiments~1 and~2.} In Experiment~1, thresholding is at
$\tau=0.5$, with edge weights drawn uniformly from $[0,1]$; the fixed
partition groups the ten vertices into pairs, one of which contains a vertex of
each part; and the thresholded power lets a path of length at most $3$ produce
its endpoint pair when all of its edges have weight at least $0.35$. Drop sinks
and edges on a cycle are as in Examples~\ref{ex:dropsinks}
and~\ref{ex:cycle}, the latter with cycles of length at most $4$.
Table~\ref{tab:exp1} gives the results by observation and density.
Experiment~2 uses three standard models (Erd\H{o}s--R\'enyi with a fixed number
of edges~\cite{erdos1959random}, Watts--Strogatz~\cite{watts1998collective},
and Barab\'asi--Albert~\cite{barabasi1999emergence}); ten further canonical
models (Price~\cite{price1965networks,price1976general}, directed
scale-free~\cite{bollobas2003directed}, Kleinberg on a
ring~\cite{kleinberg2000navigation}, stochastic
block~\cite{holland1983stochastic}, random
geometric~\cite{penrose2003random,dall2002random},
Holme--Kim~\cite{holme2002growing}, relaxed
caveman~\cite{watts1999networks,fortunato2010community},
heavy-tailed directed configuration~\cite{molloy1995critical,newman2001random},
preferential random $k$-out (NetworkX \texttt{random\_k\_out\_graph}), and
periodic square lattice); and three empirical networks (karate
club~\cite{zachary1977information}, Florentine
families~\cite{padgett1993robust,breiger1986cumulated}, and \emph{Les
Mis\'erables}~\cite{knuth1993stanford}). The empirical networks are read as
symmetric digraphs and split at random into halves, with the interaction drawn
from the edges that cross the split. Undirected models are oriented at random,
each edge reciprocated with probability $0.3$; synthetic parts have $12$
vertices and a mean total degree of $3.5$ to $5.2$.

\begin{table}[tbp]
\centering
\caption{\textbf{Experiment~1: thirteen observations.} Means over $|\mathcal I|\in\{1,2,3\}$ and $30$
realizations ($90$ instances per cell); ``ratio,'' the ratio of the means;
``att.,'' instances with $|E^+_\Phi|=N_\partial>0$; ``viol.,'' violations of
the bound among all $270$ instances of an observation.}
\label{tab:exp1}
\footnotesize
\setlength{\tabcolsep}{3.4pt}
\begin{adjustbox}{max width=\linewidth}
\begin{tabular}{@{}l*{3}{rrrr@{\hspace{10pt}}}r@{}}
\toprule
& \multicolumn{4}{c}{$p=0.15$} & \multicolumn{4}{c}{$p=0.30$} & \multicolumn{4}{c}{$p=0.45$} & \\
\cmidrule(lr){2-5}\cmidrule(lr){6-9}\cmidrule(lr){10-13}
$\Phi$ & $|E^+_\Phi|$ & $N_\partial$ & ratio & att. & $|E^+_\Phi|$ & $N_\partial$ & ratio & att. & $|E^+_\Phi|$ & $N_\partial$ & ratio & att. & viol.\\
\midrule
\multicolumn{14}{@{}l}{\emph{edgewise (reach $1$)}}\\
identity & 0.00 & 2.00 & 0.000 & 0 & 0.00 & 2.00 & 0.000 & 0 & 0.00 & 2.00 & 0.000 & 0 & 0 \\
thresholding & 0.00 & 0.97 & 0.000 & 0 & 0.00 & 0.94 & 0.000 & 0 & 0.00 & 1.13 & 0.000 & 0 & 0 \\
coarse-graining & 0.00 & 1.90 & 0.000 & 0 & 0.00 & 1.97 & 0.000 & 0 & 0.00 & 1.90 & 0.000 & 0 & 0 \\
\addlinespace
\multicolumn{14}{@{}l}{\emph{powers (reach $\ge2$)}}\\
$G^{\le2}$ & 2.46 & 4.60 & 0.534 & 0 & 4.86 & 7.16 & 0.679 & 0 & 7.02 & 9.77 & 0.719 & 0 & 0 \\
$G^{\le3}$ & 4.02 & 6.76 & 0.595 & 0 & 10.19 & 15.72 & 0.648 & 0 & 15.71 & 26.77 & 0.587 & 0 & 0 \\
$G^{\le4}$ & 4.86 & 7.99 & 0.608 & 0 & 13.14 & 23.22 & 0.566 & 0 & 20.91 & 54.31 & 0.385 & 0 & 0 \\
coarse-grained $G^{\le2}$ & 1.43 & 4.58 & 0.313 & 0 & 1.94 & 6.71 & 0.290 & 0 & 2.11 & 8.67 & 0.244 & 0 & 0 \\
coarse-grained $G^{\le3}$ & 1.84 & 6.44 & 0.286 & 0 & 3.30 & 13.79 & 0.239 & 0 & 3.77 & 24.62 & 0.153 & 0 & 0 \\
thresholded $G^{\le3}$ & 1.87 & 3.41 & 0.547 & 0 & 3.06 & 4.79 & 0.638 & 0 & 6.51 & 9.76 & 0.667 & 0 & 0 \\
\addlinespace
\multicolumn{14}{@{}l}{\emph{deletion rules that read context}}\\
drop sinks & 0.53 & 2.86 & 0.187 & 3 & 0.57 & 5.31 & 0.107 & 3 & 0.18 & 7.49 & 0.024 & 0 & 0 \\
edges on a cycle, $L{=}4$ & 0.12 & 0.37 & 0.333 & 0 & 0.11 & 0.66 & 0.169 & 0 & 0.07 & 1.07 & 0.062 & 0 & 0 \\
continuation, $L{=}2$ & 0.61 & 2.81 & 0.217 & 9 & 0.80 & 5.28 & 0.152 & 2 & 0.54 & 7.46 & 0.073 & 0 & 0 \\
continuation, $L{=}3$ & 0.84 & 2.17 & 0.390 & 15 & 1.80 & 7.91 & 0.228 & 4 & 1.77 & 16.91 & 0.104 & 1 & 0 \\
\bottomrule
\end{tabular}
\end{adjustbox}
\end{table}

\paragraph{Bridge law and shares.} The $3000$ joins that check the bridge law
use random digraphs, trees, trees with extra edges, and acyclic digraphs as
parts. The $3200$ joins that check the shares and
Proposition~\ref{prop:intervene} have two parts of $3$ to $8$ vertices or three
of $3$ to $6$, one to four interaction edges in either direction, and every
type of observation of Experiment~1; every nonempty set of interaction edges is
removed in turn.

\paragraph{Experiment~3.} The random pairs are drawn as in Experiment~1, with
four vertices per part, $p=0.3$, and one interaction edge from $s_1$ to $s_2$.
In the tree design, $s_1$ is a binary tree of depth two directed toward its
root $a$, and $s_2$ a binary tree of depth two directed away from its root $v$
(seven vertices each), so that every vertex lies within $L-1=2$ steps of the
interface; the $c$ chords are distinct ordered pairs of vertices of one part
that are not yet edges, drawn uniformly, and $c=0$ is a single graph.

\paragraph{Experiment~7.} The Erd\H{o}s--R\'enyi parts are directed $G(n,m)$
graphs, the Barab\'asi--Albert parts have attachment parameter $2$, and the
configuration parts draw in- and out-degrees with $P(k)\propto k^{-3}$.  The
part sizes are $n=10^2$, $3\times10^2$, $10^3$, $3\times10^3$, $10^4$,
$3\times10^4$, and $10^5$, with $50$ realizations for $n\le10^3$, $40$ for $3\times10^3$ and $10^4$, $30$ for
$3\times10^4$, and $20$ for $10^5$. The local search is that of
Appendix~\ref{app:computation}. The direct computation
(breadth-first searches from every vertex of $s_1$, $s_2$, and $s_1\vee s_2$)
ran on the first $50$, $50$, $50$, $40$, $40$, $10$, and $5$ realizations of
the successive sizes when the observed join had at most about $3\times10^7$
edges; at $n=10^2$ the path enumeration of Experiments~1--6 also returns the same
$E^+_\Phi$ and $N_\partial$ on $600$ instances. Local work counts the adjacency
entries scanned; direct work counts breadth-first relaxations to depth $L$,
estimated from $256$ sampled sources per graph (within $0.6\%$ of the exact
count on average). Times are medians over five realizations on an Apple M3 Max
(16 cores) with Python~3.12 and SciPy~1.17.

\paragraph{Experiment~8.} The three densities vary each generator's parameter
(for instance $n$, $2n$, $3n$ edges for Erd\H{o}s--R\'enyi, $k=2,4,6$ for
Watts--Strogatz, and $m=1,2,3$ for the Barab\'asi--Albert, Price, and Holme--Kim
models). Interaction edges are drawn from all cross pairs in both directions
(for the empirical networks, from the edges crossing the split, with $90$
realizations per $|\mathcal I|$). Edge betweenness is computed by Brandes'
algorithm~\cite{brandes2001faster}; the coarse-graining groups consecutive
triples of vertices. For ranking, $10$ pairs of parts are drawn per synthetic
family and density and $30$ per empirical network, with up to $400$ candidates
per pair, and a tie counts as a uniformly random choice among the tied
candidates. Intervals come from $2000$ bootstrap resamples. The size replicates
use synthetic parts of $12$ and $30$ vertices ($1560$ instances each).

\paragraph{Experiment~9.} The connectome is the hermaphrodite chemical-synapse
adjacency matrix of Cook et al.~\cite{cook2019whole}, with an edge $u\to v$
whenever the entry is positive, restricted to the $272$ sex-shared neurons that
they type as sensory neurons, interneurons, or motor neurons, and with their
cell classes, which follow White et al.~\cite{white1986structure}; removing
$35$ self-loops leaves $3355$ edges. The e-mail network is SNAP
email-Eu-core with its department labels~\cite{leskovec2007graph,yin2017local}
($24{,}929$ edges among $1005$ members after removing $642$ self-loops), and the
blog network is that of Adamic and Glance~\cite{adamic2005political}
($19{,}022$ hyperlinks among $758$ liberal and $732$ conservative blogs after
removing $3$ self-loops and $65$ parallel edges). The coarse-grained observation is
applied to the connectome and the e-mail network, whose data give blocks finer
than the parts. Observations are applied directly by dense
Boolean matrix products. The interventions cut, for each network and
observation, each interface (the three of largest credit in the e-mail
network), the ten interaction edges of largest credit, and the interaction edge
necessary for the most emergent edges. The data files are downloaded from their
published sources and verified by checksums.

\begin{table}[tbp]
\centering
\caption{\textbf{Experiment~9: real directed networks.} Parts and blocks are
given by the data, and c.-g.\ marks the observation coarse-grained by the
blocks; $|\mathcal I|$ counts the edges between parts. ``Top
interface'' is the largest credit of an interface (the edges from one part to
another), and ``top $10$'' the credit of the ten interaction edges of largest
credit, taken as a set, both as fractions of $|E^+_\Phi|$ (S, I, M: sensory,
inter-, and motor neurons; L, C: liberal and conservative blogs). For the e-mail
network under $G^{\le3}$, $N_\partial$ is obtained from the path-count identity
(Corollary~\ref{cor:pathcount}) without listing the paths, so its credits are
not listed.}
\label{tab:real}
\footnotesize
\setlength{\tabcolsep}{3.2pt}
\begin{adjustbox}{max width=\linewidth}
\begin{tabular}{@{}lrrrlrrrlr@{}}
\toprule
network & $n$ & parts & $|\mathcal I|$ & $\Phi$ & $|E^+_\Phi|$ & $N_\partial$ & $\Delta_W$ & top interface & top $10$\\
\midrule
\emph{C.~elegans} & 272 & 3 & 1728 & $G^{\le2}$ & 16{,}944 & 42{,}682 & 43{,}220 & 0.42 (I$\to$M) & 0.03 \\
 & & & & $G^{\le3}$ & 41{,}976 & 749{,}255 & 784{,}096 & 0.44 (I$\to$M) & 0.06 \\
 & & & & c.-g.\ $G^{\le2}$ (classes) & 3173 & 41{,}860 & 43{,}220 & 0.45 (S$\to$I) & 0.06 \\
\addlinespace
e-mail (EU core) & 1005 & 42 & 16{,}284 & $G^{\le2}$ & 288{,}577 & 1{,}330{,}677 & 1{,}341{,}903 & 0.03 & 0.01 \\
 & & & & $G^{\le3}$ & 669{,}211 & 84{,}277{,}907 & 86{,}573{,}440 & -- & -- \\
 & & & & c.-g.\ $G^{\le2}$ (depts.) & 387 & 1{,}265{,}904 & 1{,}341{,}903 & 0.05 & 0.17 \\
\addlinespace
political blogs & 1490 & 2 & 1683 & $G^{\le2}$ & 37{,}290 & 89{,}544 & 89{,}760 & 0.53 (C$\to$L) & 0.06 \\
 & & & & $G^{\le3}$ & 196{,}175 & 3{,}914{,}593 & 3{,}973{,}930 & 0.55 (L$\to$C) & 0.12 \\
\bottomrule
\end{tabular}
\end{adjustbox}
\end{table}

\phantomsection

\end{document}